\documentclass[manuscript, nonacm]{acmart}
\usepackage[T1]{fontenc}
\usepackage[utf8]{inputenc}
\usepackage{booktabs} 
\usepackage{tabulary}
\usepackage{amsthm, amsfonts, amsmath}
\usepackage{mathtools}
\usepackage{algorithm}
\usepackage{algorithmicx}
\usepackage[noend]{algpseudocode}
\usepackage{enumerate}
\usepackage{xcolor}
\usepackage{xspace}
\usepackage{subcaption}
\usepackage{nicefrac}
\usepackage[skins]{tcolorbox}
\usepackage{thmtools}
\usepackage{thm-restate}

\newtheorem{claim}{Claim}

\usepackage[textsize=small]{todonotes}

\usepackage{hyperref}
\definecolor{darkgreen}{rgb}{0,0.5,0}
\definecolor{darkblue}{rgb}{0,0,0.5}
\hypersetup{
    unicode=false,
    colorlinks=true,
    linkcolor=darkblue,
    citecolor=darkgreen,
    filecolor=magenta,
    urlcolor=cyan
}

\newcommand{\pr}[1]{\ensuremath{\text{{\bf Pr}$\left[#1\right]$}}}

\newcommand{\hide}[1]{}

\newcommand{\LOCAL}{\ensuremath{\mathsf{LOCAL}}\xspace}
\newcommand{\CONGEST}{\ensuremath{\mathsf{CONGEST}}\xspace}

\newcommand{\id}{\ensuremath{\text{id}}}

\renewcommand{\Pr}{\mathbb{P}}

\begin{document}
\title{Time-Optimal Construction of Overlay Networks
}

\author{Thorsten G\"otte}
\affiliation{
  \institution{Paderborn University}
  \streetaddress{Warburger Str. 100}
  \city{Paderborn}
  \country{Germany}}
\email{thgoette@mail.upb.de}

\author{Kristian Hinnenthal}
\affiliation{
  \institution{Paderborn University}
  \streetaddress{Warburger Str. 100}
  \city{Paderborn}
  \country{Germany}}
\email{krijan@mail.upb.de}

\author{Christian Scheideler}
\affiliation{
  \institution{Paderborn University}
  \streetaddress{Warburger Str. 100}
  \city{Paderborn}
  \country{Germany}}
\email{scheidel@mail.upb.de}

\author{Julian Werthmann}
\affiliation{
  \institution{Paderborn University}
  \streetaddress{Warburger Str. 100}
  \city{Paderborn}
  \country{Germany}}
\email{jwerth@mail.upb.de}

\begin{abstract}\normalsize
    In this paper, we show how to construct an overlay network of constant degree and diameter $O(\log n)$ in time $O(\log n)$ starting from an arbitrary weakly connected graph.
    We assume a synchronous communication network in which nodes can send messages to nodes they know the identifier of, and new connections can be established by sending node identifiers.
    If the initial network's graph is weakly connected and has constant degree, then our algorithm constructs the desired topology with each node sending and receiving only $O(\log n)$ messages in each round in time $O(\log n)$, w.h.p., which beats the currently best $O(\log^{3/2} n)$ time algorithm of [Götte et al., SIROCCO'19].
    Since the problem cannot be solved faster than by using pointer jumping for $O(\log n)$ rounds (which would even require each node to communicate $\Omega(n)$ bits), our algorithm is asymptotically optimal.
    We achieve this speedup by using short random walks to repeatedly establish random connections between the nodes that quickly reduce the conductance of the graph using an observation of [Kwok and Lau, APPROX'14].
    
    Additionally, we show how our algorithm can be used to efficiently solve graph problems in \emph{hybrid networks} [Augustine et al., SODA'20].
    Motivated by the idea that nodes possess two different modes of communication, we assume that communication of the \emph{initial} edges is unrestricted, whereas only polylogarithmically many messages can be sent over edges that have been established throughout an algorithm's execution.
    For an (undirected) graph $G$ with arbitrary degree, we show how to compute connected components, a spanning tree, and biconnected components in time $O(\log n)$, w.h.p.
    Furthermore, we show how to compute an MIS in time $O(\log d + \log \log n)$, w.h.p., where $d$ is the initial degree of $G$.
\end{abstract}

\date{}
\maketitle

\newtheorem{remark}[theorem]{Remark}
\newtheorem{fact}[theorem]{Fact}

\clearpage

\section{Introduction} 
\label{sec:intro}

Many modern distributed systems (especially those which operate via the internet) are not concerned with the physical infrastructure of the underlying network. 
Instead, these large scale distributed systems form \emph{logical networks} that are often referred to as \emph{overlay networks} or \emph{peer-to-peer networks}.
In these networks, nodes are considered as \emph{connected} if they know each other's IP-Addresses.
Practical examples for such systems are cryptocurrencies, the Internet of Things, or the Tor network.
Further examples include overlay networks like Chord~\cite{SMK+01}, Pastry~\cite{RD01}, and skip graphs~\cite{AS03}.
In this work, we consider the fundamental problem of constructing an overlay network of low diameter as fast as possible from an arbitrary initial state.
Note that $O(\log(n))$ this is the obvious lower bound for the problem:
If the nodes initially form a line, then it takes $O(\log n)$ rounds for the two endpoints to learn each other, even if every node could introduce all of its neighbors to one other in each round.

To the best of our knowledge, the first overlay construction algorithm with polylogarithmic time and communication complexity that can handle (almost) arbitrary initial states has been proposed by Angluin et al. \cite{AAC+05}. 
Here, the authors assume a weakly connected graph of initial degree $d$.
If in each round each node can send and receive at most $d$ messages, and new edges can be established by sending node identifiers, their algorithm transforms the graph into a binary search tree of depth $O(\log n)$ in time $O(d+\log^2 n)$, w.h.p.\footnote{An event holds with high probability (w.h.p.) if it holds with probability at least $1 - 1/n^c$ for an arbitrary but fixed constant c > 0.}
Since a low-depth tree can easily be transformed into many other topologies (and fundamental problems such as sorting or routing can easily be solved from such a structure), this idea has sparked a line of research investigating how quickly such overlays can be constructed.
For example, \cite{AW07} gives an $O(\log n)$ time algorithm for graphs with outdegree $1$.
If the initial degree is polylogarithmic, and nodes can send and receive a polylogarithmic number of messages, there is a deterministic $O(\log^2 n)$ time algorithm \cite{GHSS17}.
Very recently, this has been improved to $O(\log^{3/2} n)$, w.h.p. \cite{GHS19}.
However, to the best of our knowledge, there is no $O(\log(n))$-time algorithm that can construct a well-defined overly with logarithmic communication (Table \ref{tab:overview} provides an overview over the works that can be compared with our result).
In this paper, we finally close the gap and present the first algorithm that achieves these bounds, w.h.p. 
All of the previous algorithms (i.e,\cite{AAC+05,AW07,GHSS17,GHS19,GPRT20}) essentially employ the same high-level approach of \cite{AAC+05} to alternatingly group and merge so-called \emph{supernodes} (i.e., sets of nodes that act in coordination) until only a single supernode remains.
However, these supernodes need to be consolidated after being grouped with adjacent supernodes to distinguish internal from external edges. 
This consolidation step makes it difficult to improve the runtime further using this approach.
Instead, we use a radically different approach, arguably much simpler than existing solutions. It is based on classical overlay maintenance algorithms for \emph{unstructured} networks such as, for example, \cite{LS03} or \cite{GMS04}\footnote{Note that our analysis significantly differs from \cite{LS03} and \cite{GMS04} as we do not assume that nodes arrive one after the other. Instead, we assume an arbitrary initial graph of possibly small conductance.}, as well as practical libraries for overlays like JXTA \cite{OG02} or the overlay of Bitcoin. 
Instead of arranging the nodes into supernodes (and paying a price of complexity and runtime for their maintenance), we simply establish random connections between the nodes by performing short \emph{constant length} random walks.
Each node starts a small number of short random walks, connects itself with the respective endpoints, and drops all other connections. 
Then, it repeats the procedure on the newly obtained graph.
Using novel techniques by Kwok and Lau \cite{kwok2014lower} combined with elementary probabilistic arguments, we show that short random walks incrementally reduce the conductance of the graph.
Once the conductance is constant, the graph's diameter must be $O(\log n)$.
Note that such a graph can easily be transformed into many other overlay networks, such as a sorted ring, e.g., by performing a BFS and applying the algorithm of Aspnes and Wu \cite{AW07} to the BFS tree \emph{or} by using the techniques by Gmyr et al. \cite{GHSS17}

\begin{table}
\centering
\begin{tabular}{@{}llll@{}}
\hline
\bfseries Result & \bfseries Runtime & \bfseries Init. Topology & \bfseries Comm. \\
\hline
\cite{AAC+05} & $O(d+\log^2{n})$ w.h.p & Any & $O(\log(n))$\\
\cite{AW07} & $O(\log n)$ w.h.p & Outdegree $1$ & $O(\log(n))$\\
\cite{JRSST14} & $O(\log^2{n})$ w.h.p & Any & $O(n)$ \\
\cite{GHSS17} & $O(\log^2 n)$ & Any &$O(d\log(n))$  \\
\cite{GHS19} & $O(\log^{\nicefrac{3}{2}} n)$ w.h.p & Any &$O(d\log(n))$  \\
\cite{GPRT20} & $O(d\log^2 n)$ w.h.p & Any &$O(\log(n))$\\
\cite{ACC+20}&$O(\log{n})$&Line Graph& $O(\log(n))$\\
\emph{This} & $O(\log(n)$ w.h.p & Any & $O(d\log(n))$ \\
\hline
\end{tabular}
\caption{An overview of the related work. Note that $d$ denote the initial graph's degree. Communication refers to the number of messages per node and round.}
\label{tab:overview}
\end{table}

\subsection{Related Work} \label{sec:related}

The research on overlay construction is not limited to the examples given in the introduction.
Since practical overlay networks are often characterized by dynamic changes coming from \emph{churn} or \emph{adversarial behavior}, a vast amount of papers focus on reaching and maintaining a valid topology of the network in the presence of faults.
These works can be roughly categorized into two areas.
On the one hand, there are so-called \emph{self-stabilizing} overlay networks, which try to detect invalid configurations locally and recover the system into a stable state (see, e.g., \cite{FSS20} for a comprehensive survey).
However, since most solutions focus on a very general context (such as asynchronous message passing and arbitrary corrupted memory), only a few algorithms \emph{provably} achieve polylogarithmic runtimes \cite{JRSST14, BGP13}, and most have no bounds on the communication complexity.
On the other hand, there are overlay construction algorithms that explicitly use only polylogarithmic communication per node and proceed in synchronous rounds. 
In this category, we have algorithms that maintain an overlay topology under randomized or adversarial errors. 
These works focus on quickly reconfiguring the network to distribute the load evenly (under churn) or to reach an unpredictable topology (in the presence of an adversary)~\cite{DGS16, AS18, APR+15, GRS19}.
However, a common assumption is that the overlay starts in some well-defined initial state.
The work by Gilbert et al.\cite{GPRT20} combines the fast overlay construction with adversarial churn.
They present a construction algorithm that tolerates adversarial churn as long as the network always remains connected and there eventually is a period of length $\Omega(\log(n)^2)$ where no churn happens.
The exact length of this period depends on the goal topology. 
Further, there is a paper by Augustine et al.~\cite{ACC+20} that considers $\widetilde{O}(d)$-time algorithms for so-called graph realization problems.
Their goal is to construct graphs of \emph{any} given degree distributions as fast as possible.
They assume, however, that the network starts as a \emph{line}, which makes the construction of the graphs considered in this work very easy.

One of the main difficulties in designing algorithms to construct overlay networks quickly lies in the node's limited communication capabilities in a broader context. 
Therefore, our algorithm further touches on a fundamental question in designing efficient algorithms for overlay networks: How can we exploit the fact that we can (theoretically) communicate with every node in the system but are restricted to sending and receiving $O(\log(n))$ messages.
Recently, the impact of this restriction has been studied in the so-called \emph{Node-Capacitated Clique} (NCC) model \cite{AGG+19}, in which the nodes are connected as a clique and can send and receive at most $O(\log n)$ messages in each round. 
The authors present $\widetilde{O}(a)$ algorithms (where $\widetilde{O}(\cdot)$ hides polylogarithmic factors, and $a$ is the \emph{arboricity}\footnote{The arboricity of a graph is the minimum number of forests its edges can be partitioned into.} of $G$) for local problems such as MIS, matching, or coloring, a $\widetilde{O}(D+a)$ algorithm for BFS tree, and a $\widetilde{O}(1)$ algorithm for the minimum spanning tree (MST) problem.
Robinson \cite{Rob21} investigates the information the nodes need to learn to solve graph problems and derives a lower bound for constructing spanners in the NCC.
Interestingly, his result implies that spanners with constant stretch require polynomial time in the NCC and are therefore harder to compute than MSTs.
As pointed out in \cite{FHS20}, the NCC is, under certain limitations, able to simulate PRAM algorithms efficiently.
If the input graph's degree is polylogarithmic, for example, we easily obtain polylogarithmic time algorithms for (minimum) spanning forests \cite{PR02, CHL01, HZ01}.
Notably, Liu et al.~\cite{LTZ20} recently proposed an $O(\log D + \log \log_{m/n} n)$ time algorithm for computing connected components in the CRCW PRAM model, which would also likely solve overlay construction.
Assadi et al.\cite{ASW19} achieve a comparable result in the MPC model (that uses $O(n^{\delta})$ communication per node) with a runtime logarithmic in the input graph's spectral expansion.
Note that, however, the NCC, the MPC model, and PRAMs are arguably more powerful than the overlay network model considered in this paper, since nodes can reach \emph{any} other node (or, in the case of PRAMs, processors can contact arbitrary memory cells), which rules out a naive simulation that would have $\Omega(\log n)$ overhead if we aim for a runtime of $O(\log n)$.
Also, if the degree is unbounded (which is our assumption for the hybrid model), simulating PRAM algorithms, which typically have work $\Theta(m)$, becomes completely infeasible.
Furthermore, since many PRAM algorithms are very complicated, it is highly unclear whether their techniques can be applied to our model. Last, there is a hybrid network model Augustine et al. \cite{AHKSS20} that combines global (overlay) communication with classical distributed models such as \CONGEST or \LOCAL.
Here, in a single round, each node can communicate with \emph{all} its neighbors in a communication graph $G$ and addition can send and receive a limited amount messages from each node in the system.
So far, most research for hybrid networks focussed on shortest-paths problems \cite{AHKSS20, KS20, FHS20}.
For example, in general graphs APSP can be solved exactly and optimally (up to polylogarithmic factors) in time $\widetilde{O}(\sqrt{n})$, and SSSP can be computed in time $\widetilde{O}(\min\{n^{2/5}, \sqrt{D}\})$ exactly.
Whereas even an $\Omega(\sqrt{n})$ approximation for APSP takes time $\widetilde{O}(\sqrt{n})$, a constant approximation of SSSP can be computed in time $\widetilde{O}(n^\varepsilon)$ \cite{AHKSS20,KS20}.
Note that these algorithms require very high (local) communication.
If the initial graph is very sparse, then SSSP can be solved in (small) polylogarithmic time and with limited local communication, exploiting the power of the NCC~\cite{FHS20}.

\subsection{Model} \label{sec:model}

We consider overlay networks with a fixed node set $V$.
Each node $u$ has a unique \emph{identifier} $\id(u)$, which is a bit string of length $O(\log n)$, where $n = |V|$. 
Further, time proceeds in \emph{synchronous rounds}\footnote{Note that some of the algorithms can be adapted to work in an asynchronous model where a round is measured by the time it takes for the slowest message to arrive. 
Such a model (arguably) captures the heterogeneity of a P2P system with nodes and connections of varying speed and data rate more faithfully. 
If all nodes know the maximum delay of a message, they can simulate the synchronous algorithm. 
A practical downside of this approach is that the algorithm operates only as fast as the slowest part of the network. 
We go in into further details in the analysis.}.
We represent the network as a directed graph $G=(V,E)$, where there is a directed edge $(u,v) \in E$ if $u$ knows $\id(v)$.


If $u$ knows $\id(v)$ in round $i$, then it can send a message to $v$ that will be received at the beginning of round $i+1$.
New connections can be established by sending node identifiers: if $u$ sends $\id(w)$ to $v$, then $v$ can establish an edge $(v,w)$.
We restrict the size of a message to $O(\log n)$ bits, which allows a message to carry a constant number of identifiers.
Furthermore, we limit the total number of messages each node can \emph{send} and \emph{receive} in each round.
More precisely, in this paper we distinguish two different model variants:

 \textbf{NCC$_0$ model:} Each node can send and receive at most $O(\log n)$ messages in each round.
    This corresponds to the so-called NCC$_0$ model~\cite{ACC+20}, which is a variant of the general \emph{Node-Capacitated Clique} (NCC) model for overlay networks~\cite{AGG+19}.
    The bound of $O(\log n)$ is argued as a natural choice, preventing algorithms from being needlessly complicated while still ensuring scalability.
    Since this model is very general, our main algorithm will be presented in this model.
    
    
 \textbf{Hybrid model:} As in the hybrid model of Augustine et al.~\cite{AHKSS20}, we distinguish between \emph{local edges}, which are edges of the initial networks, and \emph{global edges}, which are additional edges that are established throughout an algorithm's execution.
    In each round, every node can send a single message of size $O(\log n)$ over each local edge, which corresponds to the \CONGEST model.
    Furthermore, it can send and receive a polylogarithmic number of messages over global edges.
    Note that our model corresponds to the model of \cite{AHKSS20} for \emph{local capacity} $\lambda = O(1)$ and \emph{global capacity} $\gamma = \widetilde{O}(1)$ with the difference that global edges need to be established explicitly, whereas the global network forms a clique in \cite{AHKSS20}.
    Further, whereas the algorithms presented in \cite{AHKSS20} require each node to only send and receive $O(\log n)$ messages in each round using the global network (i.e., the global capacity is $O(\log n)$), we allow polylogarithmically many messages to be sent.
    This allows us to achieve very efficient algorithms even for high initial node degrees without focusing too much on the technicalities required to achieve a global capacity of $O(\log n)$.

We assume that if (in any of these models) more messages than allowed are sent to a node, the node receives an \emph{arbitrary} subset (and the rest is simply dropped by the network).
Furthermore, we assume that every node has sufficient memory for our protocol to work correctly and every node is sufficiently fast so that it can process all messages that arrived at the beginning of round $i$ within that round\footnote{Note that for our algorithm polylogarithmic memory and local computations are sufficient.}.

\subsection{Problem Statement(s) \& Our Contribution} \label{sec:problem}

Before we formally define the problems considered in this paper, we first review some basic concepts from graph theory.
Recall that $G = (V,E)$ is a directed graph.
A node's \emph{outdegree} denotes the number of outgoing edges, i.e., the number of identifiers it stores.
Analogously, its \emph{indegree} denotes the number of incoming edges, i.e., the number of nodes that store its identifier. 
A node's \emph{degree} is the sum of its in- and outdegree, and the graph's degree is the maximum degree of any node, which we denote by $d$.
We say that a graph is \emph{weakly connected} if there is a (not necessarily directed) path between all pairs of nodes.
A graph's {\em diameter} is the maximum over all node pairs $v, w$ of the length of a shortest path between $v$ and $w$ (where we ignore the edges' directions).

Although $G$ is a directed (knowledge) graph, for the problems considered in this paper we regard $G$ as being undirected.
Our algorithms ensure that the graph can always easily be made bidirected by letting each node introduce itself to all of its neighbors.
Further, note that apart from the first of the following problems, all problems aim at finding a solution with respect to the \emph{initial} structure of (the undirected version of) $G$.

The main goal of this paper is to construct a \emph{well-formed tree}, which is a rooted tree of constant degree and diameter $O(\log n)$ that contains all nodes of $G$.
Our main result, which is presented in Section~\ref{sec:algorithm} and proven in Section~\ref{sec:analysis}, is the following.

\begin{restatable}[Main Theorem]{thm}{main}
 \label{thm:main}
    Let $G = (V,E)$ be a weakly connected directed graph with degree $O(1)$.
    There is a randomized algorithm that constructs a well-formed tree $T_G = (V,T_V)$ in $O(\log n)$ rounds, w.h.p., in the NCC$_0$ model. 
    Over the course of the algorithm, each node sends a total of at most $O(\log^2n)$ messages, w.h.p.
\end{restatable}

Since it takes time $\Omega(\log n)$ to construct a well-formed tree starting from a line even with unbounded communication, our runtime is asymptotically optimal.
We remark that we do not require the nodes to know $n$ exactly; however, they do need to know an upper bound $L \ge \log n$ on $\log n$ such that $L = O(\log n)$.

If the initial degree was $O(d)$, and the nodes were allowed to process $\Theta(d\cdot \log n)$ many messages, then our algorithm could also achieve a runtime of $O(\log n)$.
Therefore, our result directly improves upon the $O(\log^{3/2} n)$ time algorithm of \cite{GHS19}, who assume a polylogarithmic degree and allow polylogarithmic communication.
Note that a direct comparison to the model of Angluin et al.~\cite{AAC+05, AW07} is a bit difficult: 
In their model, each node can only send a \emph{single} message in each round, and, if the initial degree is $d$, there is a lower bound of $O(d + \log n)$.
It is still unclear whether our techniques could be applied to meet the lower bound of their model.
They also pose the question of whether there is an $O(\log n)$ time algorithm if each node is allowed to communicate $d$ messages; as stated above, we only answer this question affirmatively for the case that $\Theta(d \cdot \log n)$ messages can be sent and received, which might not be optimal.

\subsubsection{Applications \& Implications}
An immediate corollary of our result is that any "well-behaved" overlay of logarithmic degree and diameter (e.g., butterfly networks, path graphs, sorted rings, trees, regular expanders, De Brujin Graphs, etc.) can be constructed in $O(\log n)$ rounds, w.h.p.
These overlays can be used by distributed algorithms to common tasks like aggregation, routing, or sampling in logarithmic time.
Furthermore, we point out the following implications of this result.

\begin{enumerate}
    \item \emph{Every} monitoring problem presented in \cite{GHSS17} can be solved in time $O(\log n)$, w.h.p., instead of $O(\log^2 n)$ deterministically. 
    These problems include monitoring the graph's node and edge count, its bipartiteness, as well as the approximate and exact weight of an MST.
    \item For \cite{GRS19,APR+15,DGS16,AS18, ACC+20}, the assumption that the graph starts in well-initialized overlay can be dropped.
    \item For most algorithms that have been presented for the NCC (and hybrid networks that model the global network by the NCC) \cite{AGG+19, AHKSS20, FHS20}, the rather strong assumption that all node identifiers are known may be dropped.
    Instead, if the initial knowledge graph has degree $O(\log n)$, we can construct a butterfly network in time $O(\log n)$, which suffices for most primitives to work (note that all presented algorithms have a runtime of $\Omega(\log n)$ anyway).
\end{enumerate}
We strongly believe that our techniques could lead to networks that are highly robust against churn and DoS-attacks, at least as long as the churn is oblivious.
An adversary with full knowledge of the communication graph that can decide, which nodes join and leave the network in given round, can easily identify minimum cuts in the network and disconnect it.
If, however, the nodes fail independently and random with a certain probability, say $p$, a logarithmic sized minimum cut (of different nodes) is enough to keep the network connected w.h.p.
We touch a bit more on this topic in the end.

In Section~\ref{sec:applications}, we then give some applications of the algorithm for the hybrid model.
As already pointed out, all of the following algorithms can be performed in the hybrid network model of Augustine et al.~\cite{AHKSS20} for $\lambda = O(\log n)$ and $\gamma = \widetilde{O}(1)$, which provides a variety of novel contributions for hybrid networks\footnote{We remark that if the global network allows nodes to contact \emph{arbitrary} nodes, which is the case in the model of \cite{AHKSS20}, then some of our results can probably also be achieved by combining efficient spanner constructions with PRAM simulations.}.
For each algorithm we give a bound on the required global capacity.
Note that using more sophisticated techniques, our algorithms may very likely be optimized to require a much smaller global capacity.
We remark that \emph{all} of the following algorithms can be adapted to achieve the same runtimes in the NCC$_0$ model, if the initial degree is constant.

\subsubsection*{Connected Components} 
Here, we consider a graph $G$ that is not (necessarily) connected. For each connected component $C$ of $G$, we want to establish a well-formed tree that contains all nodes of $C$.

The section begins by presenting an adaption of our main algorithm in Section~\ref{sec:adaption} that circumvents some problems introduced by the potentially high node degrees.
As a first application of this algorithm, in Section~\ref{sec:components} we show how to establish a well-formed tree on each connected component of $G$ (if $G$ is not connected initially).

\begin{restatable}{thm}{connected}\label{thm:connected}
    Let $G = (V,E)$ be a directed graph.
    There is a randomized algorithm that constructs a well-formed tree on each connected component of (the undirected version of) $G$ in $O(\log n)$ rounds, w.h.p., in the hybrid model.
    Further, if all components have a (known) size of $O(m)$, the runtime reduces to $O(\log m +\log\log n)$ rounds, w.h.p.
    The algorithm requires global capacity $O(\log^3 n)$, w.h.p.
\end{restatable}

Here, we first need to transform the graph into a low-arboricity spanner using the efficient spanner construction of Miller et al.~\cite{MPV+15}, which was later refined by Elkin and Neiman~\cite{EN18}.
Here, each node $v \in V$ draws an exponential random variable $\delta_u$ and broadcasts it for $O(\log(n)$ rounds.
Each node keeps all edges via which it first received the value $\delta_u$ that minimizes $d_G(u,v) - \delta_u$.
We show that, if the size of each component is bounded by $m$, it suffices to observe variables $\delta_v$ smaller than $2\log(m)$ as there are $O(\log n)$ nodes that draw higher value w.h.p. The nodes that draw higher values, simply discard them.
This speeds the algorithm to $O(\log(n))$ while still producing a subgraph with few edges.
This graph can then be rearranged into a connected $O(\log n)$-degree network, allowing us the apply our main algorithm of Theorem~\ref{thm:main}.

\subsubsection*{Spanning Trees} 
Here, the goal is to compute a (not necessarily minimum) spanning tree of $G$.
In Section~\ref{sec:spanning}, we show how to obtain a spanning tree of the initial graph by "unwinding" the random walks over which the additional edges have been established.

\begin{restatable}{thm}{spanning}\label{thm:spanning}
    Let $G = (V,E)$ be a weakly connected directed graph.
    There is a randomized algorithm that constructs a spanning tree of (the undirected version of) $G$ in $O(\log n)$ rounds, w.h.p., in the hybrid model.
    The algorithm requires global capacity $O(\log^5 n)$, w.h.p.
\end{restatable}

It is unclear whether our algorithm also helps in computing an MST; it seems that in order to do so we would need different techniques.

 \subsubsection*{Biconnected Components} 
We call an undirected graph $H$ \emph{biconnected}, if every two nodes $u,v \in V$ are connected by two directed node-disjoint paths.
Intuitively, biconnected graphs are guaranteed to remain connected, even if a single node fails.
Our goal is to find the \emph{biconnected components} of $G$, which are the maximal biconnected subgraphs of $G$.
Note that \emph{cut vertices}, which are nodes whose removal increases the number of connected components, are contained in multiple biconnected components.

We show how to apply the PRAM algorithm of Tarjan and Vishkin \cite{bcmain} to compute the biconnected components of a graph to the hybrid model.
The algorithm relies on a spanning tree computation, which allows us to use Theorem~\ref{thm:spanning} to achieve a runtime of $O(\log n)$, w.h.p.

\begin{restatable}{thm}{biconnectivity} \label{thm:biconnectivity}
    Let $G = (V,E)$ be a weakly connected directed graph.
    There is a randomized algorithm that computes the biconnected components of (the undirected version of) $G$ in $O(\log n)$ rounds, w.h.p., in the hybrid model.
    Furthermore, the algorithm computes whether $G$ is biconnected, and, if not, determines its cut nodes and bridge edges.
    The algorithm requires global capacity $O(\log^5 n)$, w.h.p.
\end{restatable}

\subsubsection*{Maximal Independent Set (MIS)} 
In the MIS problem, we ask for a set $S \subseteq V$ such that (1) no two nodes in $S$ are adjacent in the initial graph $G$ and (2) every node $v \in V \setminus S$ has a neighbor in $S$. We present an efficient MIS algorithm that combines the \emph{shattering technique}~\cite{BEPSS16, Gha16} with our overlay construction algorithm to solve the MIS problem in \emph{almost} time $O(\log d)$, w.h.p.
This technique \emph{shatters} the graph into small components of undecided nodes in time $O(\log(d))$. In these components we can efficiently compute MIS solutions using a spanning tree of depth $O(\log(\log(n)))$ which we can compute in $O(\log(\log(n)))$ rounds.
This leads to an $O(\log d + \log \log n)$ time algorithm, where $d$ is the initial graph's degree.

\begin{restatable}{thm}{mis} \label{thm:mis}
    Let $G = (V,E)$ be a weakly connected directed graph.
    There is a randomized algorithm that computes an MIS of $G$ in $O(\log d + \log \log n)$ rounds, w.h.p., in the hybrid model.
    The algorithm requires global capacity $O(\log^3 n)$, w.h.p.
\end{restatable}



\subsection{Mathematical Preliminaries}
Before we give the description of our main algorithm, we introduce some notions from probability and graph theory that we will frequently use throughout the remainder of this paper.
First, we heavily use a well-known Chernoff Bound, which is a standard tool for the analysis of distributed algorithms.
In particular, we will use the following version:

\setcounter{theorem}{5}

\begin{lemma}[Chernoff Bound]
	\label{lem:chernoffbound}
	Let $X = \sum_{i=1}^n X_i$ for independent random variables $X_i \in \{0,1\}$ and $\mathbb{E}(X) \leq \mu_H$ and $\delta \geq 1$.
		$$
		\mathbb{P}\big(X > (1 \!+\! \delta) \mu_H\big) \leq \exp\big(\!\!-\!\frac{{\delta\mu_H}}{3}\big),
	$$ 
	 Similarly, for $\mathbb{E}(X) \geq \mu_L$ and $0 \leq \delta \leq 1$ we have
	$$\mathbb{P}\big(X < (1 \!-\! \delta) \mu_L\big) \leq \exp\big(\!\!-\!\frac{{\delta^2\mu_L,}}{2}\big).$$	
\end{lemma}

Furthermore, our analysis will heavily rely on the (small-set) conductance of the communication graph.
The conductance of set $S \subset V$ is the ratio of its outgoing edges and its size $|S|$.
The conductance $\Phi(G)$ of a graph $G$ is the minimal conductance of every subset.
More precisely, we will need a more generalized notion of small-set conductance that only observes sets of a certain size.
Formally, the small-set conductance is defined as follows:
\begin{definition}[Small-Set Conductance]
    Let $G := (V,E)$ be a connected $\Delta$-regular graph and $S \subset V$ with $|S| \leq \frac{|V|}{2}$ be any subset of $G$ with at most half its nodes.
    Then, the conductance $\Phi(S) \in (0,1)$ of $S$ is defined as follows:
    $$
        \Phi(S) := \frac{|\{(v,w) \in E \,|\, v \in S, w \not\in S\}|}{\Delta|S|}
    $$
    For a parameter $\delta \in (0,1)$ the (small-set) conductance of $G$ is then defined as:
    $$
        \Phi_\delta(G) := \min_{S \subset V, |S| \leq \frac{\delta|V|}{2}} \Phi(S) 
    $$
\end{definition}
\begin{remark}
For $\delta=1$, we call $\Phi(G) := \Phi_1(G)$ simply the conductance of $G$.
\end{remark}

\section{The Overlay Construction Algorithm} 
\label{sec:algorithm}
In this section, we present our algorithm to construct a well-formed tree in time $O(\log n)$, w.h.p., and give an overview of the proof to establish the correctness of Theorem~\ref{thm:main}.
To the best of our knowledge, our approach is different from \emph{all} previous algorithms for our problem \cite{AS03,AW07,GHSS17,GHS19} in that it does \emph{not} use any form of clustering to contract large portions of the graph into supernodes. 
From a high level, our algorithm progresses through $O(\log n)$ graph evolutions, where the next graph is obtained by establishing random edges on the current graph.
More precisely, each node of a graph simply starts few random walks of constant length and connects itself with the respective endpoints.
The next graph only contains the newly established edges.
We will show that after $O(\log n)$ iterations of this simple procedure, we reach a graph that has diameter $O(\log n)$.

One can easily verify that this strategy does not trivially work on any graph, as the graph's degree distributions and other properties significantly impact the distribution of random walks.
However, as it turns out, we only need to ensure that the initial graph has some \emph{nice} properties to obtain well-behaved random walks.
More precisely, throughout our algorithm, we maintain that the graph is \emph{benign}, which we define as follows.
\begin{definition}[Benign Graphs]
\label{def:benign}
    Let $G := (V,E)$ be a directed graph and $\Delta,\Lambda = \Omega(\log n)$ be two arbitrary values (with big enough constants hidden by the $\Omega$-Notation).
    Then, we call $G$ benign if and only if it has the following three properties:
    \begin{enumerate}
        \item \textbf{($G$ is $\Delta$-regular)} Every node $v \in V$ has exactly $\Delta$ in- and outgoing edges (which may be self-loops).
        \item \textbf{($G$ is lazy)} 
        Every node $v \in V$ has at least $\nicefrac{1}{2}\Delta$ self-loops.
        \item \textbf{($G$ has a $\Lambda$-sized minimum cut)}
        Every cut $c(V,\overline{V})$ has at least $\Lambda$ edges.
    \end{enumerate}{}
\end{definition}
The properties of benign graphs are carefully chosen to be as weak as possible while still ensuring the correct execution of our algorithm. 
A degree of $\Delta = \Omega(\log n)$ is necessary to keep the graph connected. 
If we only had a constant degree, a standard result from random graphs implies that w.h.p. there would be nodes disconnected from the graph when sampling new neighbors.
If the graphs were not lazy, many theorems from the analysis of Markov chains would not hold as the graph could be bipartite, which would greatly complicate the analysis.
This assumption only slows down random walks by a factor of $2$.
Lastly, the $\Lambda$-sized cut ensures that the graph becomes more densely connected in each evolution,  w.h.p.
In fact, with constant-sized cuts, we cannot easily ensure this property when using random walks of constant length. 

\subsection{Algorithm Description}
We will now describe the algorithm in more detail.
Recall that throughout this section, we will assume the NCC$_0$ model, which means that each node can send and receive $O(\log n)$ distinct messages. 
Further, we assume for simplicity that the initial graph has at most a constant maximal degree $d = O(1)$ and is connected.\footnote{With more complex prepossessing, this assumption can be removed and $d$ can be raised to $O(\log n)$. However, to concentrate on novel aspects of our algorithm, we make this simplification here.}

Besides the initial set of edges, the algorithm has four input parameters $\ell,\Delta,\Lambda$, and $L$ that are known to all nodes.
Recall that $L = O(\log n)$ is an upper bound on $\log n$.
The value $\ell = \Omega(1)$ denotes the length of the random walks, $\Delta = O(\log n)$ is the desired degree, and $\Lambda = O(\log n)$ denotes the size of the minimum cut.
All of these parameters are \emph{tunable} and the hidden constants need to be chosen big enough for the algorithm to succeed w.h.p.
We discuss this in more detail in the analysis.

Before the first evolution, we need to prepare the initial communication graph to comply with these parameters, i.e., we must turn it into a benign graph.
Since the input graph has a maximal degree of $d = O(1)$, this is quite simple as we can assume $2d\Lambda \leq \Delta = O(\log n)$.
Given this assumption, the graph can be turned benign in 2 steps.
First, all edges are copied $\Lambda$ times to obtain the desired minimum cut. 
After this step, each node has at most $d\Lambda$ edges to other nodes.
Then, each node adds self-loops until its degree is $\Delta$ and each node has $\frac{\Delta}{2}$ self-loops.
As we chose $2d\Lambda \leq \Delta$, this is always possible.

Let now $G_0 = (V, E_0)$ be the resulting benign graph.
The algorithm proceeds in iterations $1, \dots, L$.
In each iteration, a new communication graph $G_i = (V, E_i)$ is created through sampling $\frac{\Delta}{8}$ new neighbors via random walks of length $\ell$.
Each node $v \in V$ creates $\frac{\Delta}{8}$ messages containing its own identifier, which we call \emph{tokens}.
Each token is randomly forwarded for $\ell$ rounds in $G_i$. 
More precisely, each node that receives a token picks one of its incident edges in $G_i$ uniformly at random and sends the token to the corresponding node.\footnote{We will show that each node only sends and receives at most $O(\log n)$ tokens in each round, w.h.p.}
If $v$ receives less than $\frac{3}{8}\Delta$ tokens after $\ell$ steps, it sends its own identifier back to all the tokens' origins to create a bidirected edge. 
Otherwise, it picks $\frac{3}{8}\Delta$ tokens at random (without replacement)\footnote{We will see that this case does not occur w.h.p.}.
Since the origin's identifier is stored in the token, both cases can be handled in one communication round.
Finally, each node adds self-loops until its degree is $\Delta$ again.
The whole procedure is given in Figure \ref{fig:pseudocode} as the method \textsc{CreateExpander}($G_0,\ell,\Delta,\Lambda,L$).
The subroutine \textsc{MakeBenign}($G_0,\ell,\Delta,\Lambda$) add edges and self-loops to make the graph comply to Definition \ref{def:benign}.
\begin{figure}[ht]
\begin{tcolorbox}[enhanced,
    ,drop shadow southwest]
\underline{\textsc{CreateExpander}}($G_0,\ell,\Delta,\Lambda,L$):\\
Each node $v \in V$ executes:
\begin{enumerate}
    \item $E_0 \longleftarrow $\textsc{MakeBenign}($G_0,\ell,\Delta,\Lambda$)
    \item For $i=0, \dots, L$:
    \begin{enumerate}
        \item Create $\nicefrac{\Delta}{8}$ tokens that contain $v$'s identifier and store them in $T_0$.
        \item For $j=1, \dots, \ell$:
        \begin{enumerate}
            \item[] Independently send each token from $T_{j-1}$ along a random incident edge in $G_i = (V,E_i)$.
            \item[] Store all received token in the buffer $T_j$.
        \end{enumerate}
        \item Pick (up to) $\nicefrac{3\Delta}{8}$ tokens $w_1, \dots, w_{\Delta'}$ from $T_\ell$\\ without replacement.
        \item Create edges $E_{i+1} := \{\{v,w_1\}, \dots, \{v,w_{\Delta'}\}\}$ by sending $v$'s identifier to each $w_j$.
        \item Add self-loops $\{v,v\}$ to $E_{i+1}$ until $|E_{i+1}|=\Delta$
    \end{enumerate}
\end{enumerate}
\end{tcolorbox}
\caption{Pseudocode for our main algorithm.}
\label{fig:pseudocode}
\end{figure}
Our main observation is that after $L = O(\log n)$ iterations, the resulting graph $G_L$ has constant conductance, w.h.p., which implies that its diameter is $O(\log n)$.
Furthermore, the degree of $G_L$ is $O(\log n)$.
To obtain a well-formed tree $T_G$, we first perform a BFS on $G_L$ starting from the node with lowest identifier\footnote{Since a node cannot locally check whether it has the lowest identifier, the implementation of this step is slightly more complex: Every node simultaneously floods the graph with a token message that contains its identifier. Every node that receives one or more tokens only forwards the token with the lowest identifier. Since the graph's diameter is $O(\log n)$, all nodes know the lowest identifier after this time}.
This requires time $O(\log n)$ and gives us a rooted tree $T$ with degree and diameter $O(\log n)$.
To transform this tree into a well-formed tree, we perform the \emph{merging step} of the algorithm of \cite[Theorem 2]{GHSS17}.
From a high level, the algorithm first transforms $T$ into a constant-degree \emph{child-sibling tree}~\cite{AW07}, in which each node arranges its children as a path and only keeps an edge to one of them.
Using the \emph{Euler tour technique} (see, e.g., \cite{bcmain}), this tree is then transformed into a rooted tree of constant degree and depth $O(\log n)$ in time $O(\log n)$.
This tree is our desired well-formed tree $T_G$, which concludes the algorithm.

\subsection{Analysis Overview}
Before we go into the proof's intricate details, let us first prove that during the execution of the algorithm all messages are successfully sent.
Remember that we assume the nodes to have a capacity of $O(\log n)$ and thus a node can only send and receive $O(\log n)$ messages as excess messages are dropped.
That means, in order to prove that no message is dropped, we must show that no node receives more than $O(\log n)$ random walk tokens in a single step.

For the proof, we observe a well known fact about the distribution of random walks that has been independently shown in \cite{DGS16}, \cite{CFSV19} and \cite{DSMPU13}:
\begin{lemma}[Shown in \cite{DGS16,CFSV19,DSMPU13}]
For a node $v \in V$ and an integer $t$ let $X(v,t)$ be the random variable that denotes the number of token at node $v$ in round $t$.
Then, it holds $\pr{X(v,t) \geq \frac{3\Delta}{8}} \leq \frac{1}{e^\frac{\Delta}{8}}$.
\end{lemma}
The lemma follows from the fact that each node receives $\frac{\Delta}{8}$ tokens  in expectation given that all neighbors received $\frac{\Delta}{8}$ tokens in the previous round. 
This holds because $G_i$ is regular.
Since all nodes start with $\frac{\Delta}{8}$ tokens, the lemma follows inductively.
Since all walks are independent, a simple application the Chernoff Bound yields the result.
Note that this Lemma also directly implies that, w.h.p., all random walks create an edge as every possible endpoint receives less than $\frac{3\Delta}{8}$ token and therefore replies to all of them. 
In the remainder the analysis we will implicitly condition all random choices on these facts. 

The main challenge of our analysis is to show that after $L \in O(\log(n))$ evolutions, the final graph $G_{L}$ has diameter of $O(\log(n))$.
Given this fact, the technique from \cite[Theorem 2]{GHSS17} transforms $G_L$ into well-formed tree in $O(\log(n))$ rounds.
Thus, our analysis will focus on showing that $G_L$ has logarithmic diameter.
To do so, we will perform an induction over the sequence of graphs $\mathcal{G} := G_1, \ldots, G_L$.
Our main insight is that --- given the communication graph is benign --- we can use short random walks of \emph{constant} length to iteratively increase the graph's conductance until we reach a graph of low diameter. In particular, we show that the graph's conductance is strictly increasing by a factor $\Omega(\sqrt{\ell})$ from $G_i$ to $G_{i+1}$ if $G_i$ is benign, i.e,.
\begin{lemma} 
    Let $G_i$ and $G_{i+1}$ be the graphs created in iteration $i$ and $i+1$ respectively and assume that $G_i$ is benign with a minimum cut of at least $\Lambda \geq 640$. 
    Then, it holds
    \begin{equation}
        \Phi(G_{i+1}) \geq \min\{ \frac{1}{2}, \frac{1}{640}\sqrt{\ell}\Phi(G_i)\}
    \end{equation}
    In particular, for any $\ell \geq 2 \cdot 640^2$, it holds 
    \begin{equation}
        \Phi(G_{i+1}) \geq \min\{ \frac{1}{2}, 2\cdot\Phi(G_i)\}
    \end{equation}
\end{lemma}
Intuitively, this makes sense as the conductance is a graph property that measures how well-connected a graph is and --- since the random walks monotonically converge to the uniform distribution ---  the newly sampled edges can only increase the graph's connectivity. 

To be precise, our main argument is the fact that random walks of length $\ell$ are distributed according to $G_i^{\ell}$, where $G^{\ell}$ is the $\ell^{th}$ power of the random walk matrix for $G_i$. 
Since, for the most part, we will consider the evolution from $G_i$ to $G_{i+1}$, we will refer to $\Phi(G_i)$ and $\Phi(G_i^\ell)$ simply as $\Phi$ and $\Phi_\ell$.
In particular, if we consider a subset $S \subset V$, then $\Phi_\ell$ denotes the probability that a random walk ends outside of the subset after $\ell$ steps, in which case the corresponding node creates an edge outgoing of the subset.
Since we ensure that the total number of edges in a set stays constant, this creation of an outgoing edge increases the set's conductance in expectation.
Thus, we can show the following lemma:
\begin{lemma}
\label{lemma:exp_conductance_2}
Given that each node starts $\frac{\Delta}{8}$ tokens, which all create an edge, it holds:
$$
        \mathbb{E}\left[ \Phi_{i+1}(S) \right] \geq \frac{\Phi_\ell}{8}
$$
\end{lemma}
Therefore, a lower bound on $\Phi_\ell$ gives us a lower bound on the expected conductance of the newly sampled graph.
However, the standard Cheeger inequality (see, e.g., \cite{Sin12} for an overview) that is most commonly used to bound a graph's conductance with the help of the graph's eigenvalues does not help us in deriving a meaningful lower bound for $\Phi_\ell$.
In particular, it only states that $\Phi_\ell = \Theta(\ell\Phi^{2})$.
Thus, it only provides a useful bound if $\ell = \Omega(\Phi^{-1})$, which is too big for our purposes, as $\Omega(\Phi^{-1})$ is only constant if $\Phi$ is already constant. 
More recent Cheeger inequalities shown in \cite{LGT11} relate the conductance of smaller subsets to higher eigenvalues of the random walk matrix.
On the first glance, this seems to be helpful, as one could use these to show that at least the small sets start to be more densely connected and then, inductively, continue the argument.
Still, even with this approach, constant length walks are out of the question as the new Cheeger inequalities introduce an additional tight $O(\log n)$ factor in the approximation for these small sets.
Thus, the random walks would need to be of length $\Omega(\log n)$, which is still too much to achieve our bounds.
Instead, we use the following result by Kwok and Lau~\cite{kwok2014lower}, which states that every $\Phi_\ell$ improves even for constant values of $\ell$. 
It holds that:
\begin{lemma}[Conductance of $G^\ell$, Based on Theorem $1$ in \cite{kwok2014lower}]
    Let $G = (V,E)$ be any connected $\Delta$-regular lazy graph with conductance $\Phi$ and let $G^\ell$ be its $\ell^{\text{th}}$ power.
    For a set $S \subset G$ define $\Phi_\ell(S)$ as the conductance of $S$ in $G^\ell$.
    Then, it holds:
    $$
        \frac{1}{2} \geq \Phi_\ell(S) \geq \max\left\{\frac{1}{40}\sqrt{\ell}\Phi, \, \Phi(S)\right\}  
    $$
\end{lemma}
Given this bound, we can show that benign graphs indeed increase their (expected) conductance from iteration to iteration.

However, this fact alone is not enough to finalize the proof.
Recall that we need to show that \emph{every} subset has a conductance of $O(\sqrt{\ell}\Phi(G_i))$ in $G_{i+1}$ in order to prove that $\Phi(G_{i+1}) = \Omega(\sqrt{\ell}\Phi(G_i))$. 
Since there are exponentially many subsets, a bound in the magnitude of $o(n^{-c})$ is not sufficient.
Now recall that, since the number of edges within $S$ is unchanged, the rising conductance implies that the number of outgoing edges of each set $S$ rises by $\Omega(\sqrt{\ell})$ (or reaches $\Theta(\Delta |S|)$).
Thus, given that $S$ has $\alpha\Lambda$ outgoing edges, the value of $\Phi_{G_{i+1}}(S)$ is concentrated around its expectation with probability $e^{-\Omega{\alpha\Lambda}}$.
This follows from the Chernoff bound and the fact that the random walks are quasi independent\footnote{Technically, they are not independent since in the last step we draw from all tokens without replacement.
However, since we condition on all nodes receiving less than $\frac{3\Delta}{8}$, we can observe the experiment where each node start $\frac{\Delta}{8}$ independent random walks and connects with the endpoints. Here, the Chernoff bound holds.}.
This is then used to derive the \emph{high probability} bound.
By a celebrated result of Karger, the number of subsets with $\alpha\Lambda$ outgoing edges can be bounded by $O(n^{2\alpha})$ \cite{karger2000minimum}. 
Thus, for a big enough $\Lambda$, a bound of $e^{-\Omega{(\alpha\Lambda)}}$ is enough to show all sets increase their conductance.

Since all the arguments from before only hold if $G_i$ is benign, we must additionally make sure that each graph in $\mathcal{G} := G_1, \ldots, G_L$ is indeed begnin.
As before, we prove this step by step and show that $G_{i+1}$ is benign given that $G_i$ is benign.
\begin{lemma}
    Let $G_i$ and $G_{i+1}$ be the graphs created in iteration $i$ and $i+1$ respectively and assume that $G_i$ is benign. Then, w.h.p., the graph $G_{i+1}$ is also benign.
\end{lemma}
While the regularity and the laziness follow directly from observing the algorithm, the minimum is cut trickier.
Here, we need to use an argument similar to the one used for proving the conduction: We show that every subset (in expectation) has cut of size $2\Lambda$ if we choose $\ell$ big enough. This follows again from a lemma in \cite{kwok2014lower} and a connection between the minimum cut and the small set expansion.
Finally, we again use the Chernoff bound in conjunction with Karger's lemma to proof that all cuts are bigger than $\Delta$ w.h.p.
 
To round up the analysis, we now only need to show that after $O(\log n)$ iterations, the graph has a diameter of $O(\log n)$.
For this, we need two more facts:.
First, we observe the worst possible initial conductance.
\begin{lemma}[Minimum conductance]
    Let $G := (V,E)$ be any connected graph, then $\Phi(G) \geq \frac{1}{\Delta n}$.
\end{lemma}
The former follows from the fact that in the worst case there is an $O(n)$-sized set that is connected to the remaining graph by a single edge.
Second, we observe that a constant conductance implies a logarithmic diameter if the graph is regular. It holds:
\begin{lemma}[High Conductance implies Low Diameter]
    Let $G := (V,E)$ be any $\Delta$-regular graph with conductance $\Phi$, then the diameter of $G$ is at most $O(\Phi^{-1}\log n)$.
\end{lemma}
This can be verified, e.g., by inductively observing any the neighborhood for any two nodes $v,w \in V$.
Let $N_i(v)$ contain all nodes in distance $i$ to $v$. 
As long as $|N_i(v)|\leq \frac{n}{2}$, there are at least $\Phi\frac{\Delta}{2}|N_i(v)|$ edges coming out of $N_i(v)$ and since $G$ is $\Delta$-regular, the set $N_{i+1}(v) := N_i(v) \cup N(N_i(v))$ is of size at least $(1+\Phi)|N_i(v)|$. Thus, after $I \in  O(\Phi^{-1}\log(n))$ iterations we must reach more than $\frac{n}{2}$ nodes. 
Since the same holds for $N_I(w)$ there must be an node $u \in N_I(v) \cap N_I(w)$. This implies a diameter of $2I \in O(\Phi^{-1}\log(n))$ 

Given that after every iteration the graph's conductance increases by a factor $O(\sqrt{\ell})$ w.h.p, a simple union bound tells us, as long as we consider $o(n)$ iterations, the conductance is increases in every iteration w.h.p.
Thus, after $O\left(\frac{\log \Phi}{\log\ell}\right) = O(\log n)$ iterations, the most recent graph must have constant conductance, w.h.p. 
Therefore, since each iteration lasts only $\ell = O(1)$ rounds, after $O(\log n)$ rounds the graph has a constant conductance and thus logarithmic diameter.
This concludes the analysis.

\subsection{Analysis of \textsc{CreateExpander}}
\label{sec:analysis}
We now provide a more detailed analysis
Before we go into the proof's intricate details, let us first prove that all messages are successfully sent during the execution of the algorithm.
Remember that we assume the nodes to have a capacity of $O(\log n)$ and thus, a node can only send and receive $O(\log n)$ messages as excess messages are dropped arbitrarily.
That means, in order to prove that no message is dropped, we must show that no node receives more than $O(\log n)$ random walk tokens in a single step.
However, this is a well known fact about the distribution of random walks:
\begin{lemma}[Shown in \cite{DGS16,CFSV19,DSMPU13}]
\label{lemma:expected_congestion}
For a node $v \in V$ and an integer $t$ let $X(v,t)$ be the random variable that denotes the number of token at node $v$ in round $t$.
Then, it holds $\Pr{X(v,t) \geq \frac{3\Delta}{8}} \leq e^{-\frac{\Delta}{8}}$.
\end{lemma}
The lemma follows from the fact that each node receives $\frac{\Delta}{8}$ tokens in expectation given that all neighbors received $\frac{\Delta}{8}$ tokens in the previous round. 
This holds because $G_i$ is regular.
Since all nodes start with $\frac{\Delta}{8}$ tokens, the lemma follows inductively.
Since all walks are independent, a simple application of the Chernoff Bound yields the result.
Note that this Lemma also directly implies that, w.h.p., all random walks create an edge as every possible endpoint receives less than $\frac{3\Delta}{8}$ token and therefore replies to all of them. 
In the remainder of the analysis, we will implicitly condition all random choices on these facts. 

The main challenge of our analysis is to show that after $L \in O(\log(n))$ evolutions, the final graph $G_{L}$ has diameter of $O(\log(n))$.
Given this fact, the technique from \cite[Theorem 2]{GHSS17} transforms $G_L$ into well-formed tree in $O(\log(n))$ rounds.
Thus, our analysis will focus on showing that $G_L$ has a logarithmic diameter.
To do so, we will perform an induction over the sequence of graphs $\mathcal{G}:= G_1, \ldots, G_L$.
Our main insight is that --- given the communication graph is benign --- we can use short random walks of \emph{constant} length to iteratively increase the graph's conductance until we reach a graph of low diameter. In particular, we show that the graph's conductance is strictly increasing by a factor $\Omega(\sqrt{\ell})$ from $G_i$ to $G_{i+1}$ if $G_i$ is benign, i.e.,
\begin{lemma} 
\label{lemma:induction_step}
    Let $G_i$ and $G_{i+1}$ be the graphs created in iteration $i$ and $i+1$ respectively and assume that $G_i$ is benign with a minimum cut of at least $\Lambda \geq 640$. 
    Then, it holds
    \begin{equation}
        \Phi_{G_{i+1}} \geq \min\left\{ \frac{1}{2}, \frac{1}{640}\sqrt{\ell}\Phi_{G_i}\right\}
    \end{equation}
    In particular, for any $\ell \geq 2 \cdot 640^2$, it holds 
    \begin{equation}
        \Phi_{G_{i+1}} \geq \min\left\{ \frac{1}{2}, 2\cdot\Phi_{G_i}\right\}
    \end{equation}
\end{lemma}
Intuitively, this makes sense as the conductance is a graph property that measures how well-connected a graph is and --- since the random walks monotonically converge to the uniform distribution ---  the newly sampled edges can only increase the graph's connectivity. 

Our first observation is the fact that random walks of length $\ell$ are distributed according to $1$-step walks in $G_i^{\ell}$.
In particular, if we consider a subset $S \subset V$ and pick a node $v \in S$ uniformly at random, then $\Phi_{G_i^\ell}(S)$ denotes the probability that a random walk started at $v$ ends outside of the subset after $\ell$ steps.
In this case, $v$ creates an edge to some node in $V \setminus S$.
Since we ensure that the total number of edges in a set stays constant, creating such an outgoing edge increases the set's conductance in expectation.
Thus, we can show the following lemma:
\begin{lemma}
\label{lemma:exp_conductance_2}
Let $\Phi_{G_{i+1}}(S)$ be the conductance of set $S \subset V$ in $G_{i+1}$.
Given that each node starts $\frac{\Delta}{8}$ tokens, which all create an edge, it holds:
$$
        \mathbb{E}\left[ \Phi_{G_{i+1}}(S) \right] \geq \frac{\Phi_{G^\ell_{i+1}}}{8}
$$
\end{lemma}

Before we can prove the lemma, we need some auxiliary lemmas and definitions.
For $v,w \in V$ let $X^\ell_v(w,G)$ be indicator for the event that an $\ell$-step random walk in $G$ which started in $v$ ends in $w$. Analogously, let $X^1_v(w,G^\ell)$ be the probability that a $1$-step random walk in $G^\ell$ which started in $v$ ends in $w$. 
If we consider a fixed node $v$ that is clear from the context, we may drop the subscript and write $X^1(w,G^\ell)$ instead.
Further, let $P_\ell(v,w)$ is the exact number of walks of length $\ell$ between $v$ and $w$ in $G$.

\begin{lemma}
    The probability to move from $v$ to $w$ in $G^\ell$ is given by:
\begin{align}
    \pr{X^1(w,G^\ell)} = \pr{X^\ell(w,G)} = \frac{P_\ell(v,w)}{\Delta^\ell}
\end{align}
\end{lemma}
\begin{proof}
    The statement can be proved via an induction over $\ell$, the length of the walk.
    \begin{itemize}
        \item[\textbf{(IB)}] For the base case we need to show that $1$-step random walk in $G$ is equivalent to picking an outgoing edge in $G^1 := G$ uniformly at random. This follows trivially from the very definition of a random walk.
        \item[\textbf{(IS)}] Now suppose that performing an $(\ell-1)$-step random walk in $G$ is equivalent to performing a $1$-step walk in $G^{\ell-1}$.
        Consider a node $w \in V$ and let $N_w$ denote its neighbors in $G$ \emph{and} $w$itself.
        By the law of total probability, it holds:
        \begin{align}
            \pr{X^\ell(w,G) = 1} := \sum_{u \in N_w} \pr{X^{\ell-1}(u,G) = 1}\pr{X^\ell(w,G) = 1 \mid X^{\ell-1}(u,G) = 1 } 
        \end{align}
        Using the \emph{inductions hypothesis} we we can substitute $\pr{X^{\ell-1}(u,G) = 1}$ for \pr{X^{1}(u,G^{\ell-1}) = 1} and get: 
        \begin{align}
            \pr{X^\ell(w,G) = 1} &:= \sum_{u \in N_w} \pr{X^{\ell-1}(u,G) = 1}\pr{X^\ell(w,G) = 1 \mid X^{\ell-1}(u,G) = 1 }\\
            &= \sum_{u \in N_w}  \frac{P_{\ell-1}(v,u)}{\Delta^{\ell-1}}\pr{X^\ell(w,G) = 1 \mid X^{\ell-1}(u,G) = 1 }
        \end{align}
        Recall that $G$ is a multigraph and there can be more than one edge between each $u$ and $w$.
        Thus, let now $e(u,w)$ denote the number of edges between $u$ and $w$ for every $u \in N_w$. 
        Since we defined that $w \in N_w$, the value $e(w,w)$ counts $w$'s self-loops.
        Since $G$ is $\Delta$-regular, the probability that a random walk at node $u$ moves to $w$ is exactly $\frac{e(u,w)}{\Delta}$.
        Back in the formula, we get:
        \begin{align}
            \pr{X^\ell(w,G) = 1} &= \sum_{u \in N_w}  \frac{P_{\ell-1}(v,u)}{\Delta^{\ell-1}}\pr{X^\ell(w,G) = 1 \mid X^{\ell-1}(u,G) = 1 }\\
            &= \sum_{u \in N_w}  \frac{P_{\ell-1}(v,u)}{\Delta^{\ell-1}}\frac{e(u,w)}{\Delta}
            = \frac{1}{\Delta^\ell}\sum_{u \in N_w}  {P_{\ell-1}(v,u)}\cdot{e(u,w)}
        \end{align}
        Finally, note that $\sum_{u \in N_w}  {P_{\ell-1}(v,u)}\cdot{e(u,w)}$ counts all paths of length exactly $\ell$ from $v$ to $w$ in $G$. 
        This follows because each path $P := (e_1, \ldots, e_\ell)$ from $u$ to $w$ can be decomposed into a path $P' := (e_1, \ldots, e_{\ell-1})$ of length $\ell-1$ to some neighbor of $w$ (or $w$ itself) and the final edge (or self-loop) $e_\ell$. 
        Thus, it follows that:
        \begin{align}
            \pr{X^\ell(w,G) = 1} &= \frac{1}{\Delta^\ell}\sum_{u \in N_w}  {P_{\ell-1}(v,u)}\cdot{e(u,w)} = \frac{P_{\ell}(v,w)}{\Delta^\ell} = \pr{X^1(w,G^\ell) = 1}
        \end{align}
        This was to be shown.
    \end{itemize}
\end{proof}
Next, we consider the random walks that start in a given set $S$ and end outside of it (i.e., the random walks used to create connections outside of $S$).
Given our definition from above, we can show that the number of walks that end outside of $S$ depends on the size of $s$, the degree $\Delta$, and $\Phi_{G_i^\ell}$:
\begin{lemma}[Expected Conductance]
\label{lemma:exp_conductance}
    Let $G$ be a $\Delta$-regular graph and $S \subset V$ be a any subset of nodes with $|S| := s \leq \frac{n}{2}$ and suppose each node in $S$ starts $\frac{\Delta}{8}$ random walks.
    Let $\mathcal{Y}_S$ count the $\ell$-step random walks that start at some node in $v \in S$ and end at some node $w \in V \setminus S$. Then it holds:
    $$
        \mathbb{E}\left[ \mathcal{Y}_S \right] := \frac{\Delta s}{8}\Phi_{G_i^\ell}(S)
    $$
    \end{lemma}
    \begin{proof}
    First, we observe that we can $\mathcal{Y}_S$ as the sum of binary random variables for each walk.
    For each $v_i \in S$ let $Y_i^1, \dots, Y_i^{d}$ be indicator variables that denote if a token started by $v_i$ ended in $\overline{S} := V\setminus S$ after $\ell$ steps.
    Given this definition, we see that
    \begin{align}
        \mathcal{Y}_S := \sum_{i=1}^{s}\sum_{j=1}^\frac{\Delta}{8} Y_i^j.
    \end{align}
    Recall that an $\ell$-step random walk in $G_i$ corresponds to a $1$-step random walk in $G^\ell_i$.
    This means for each of its $\frac{\Delta}{8}$ tokens node $v_j$ picks one of its outgoing edges in $G^\ell_i$ uniformly at random and sends the token along this edge (which corresponds to an $\ell$-step walk).
    For ease of notation, let $\mathbb{O}_j$ be the number of edges of node $v_j \in S$ in $G^\ell_i$ where the other endpoint is not in $S$.
    Now consider the $k^{th}$ random walk started at $v_j$ and observe $Y_j^k$. 
    Note that it holds:
    \begin{align}
    \label{eqn:v_i_prob}
        \mathbb{E}[Y_k^j(t)] &= \sum_{w \in \overline{S}} \pr{X^1_{v_j}(w,G^\ell)} \cdot \mathbb{E}[Y_k^j(t) \mid X^1_{v_j}(w,G^\ell)]\\
        &= \sum_{w \in \overline{S}} \pr{X^1_{v_j}(w,G^\ell)} = \sum_{w \in \overline{S}} \frac{P_\ell(v_j,w)}{\Delta^\ell} =\frac{\mathbb{O}_j}{\Delta^{\ell}}
   \end{align}
    Here, the denominator $\Delta^{\ell}$ comes from the fact that $G_i^{\ell}$ is $\Delta^{\ell}$-regular.
    
    Let $\mathbb{O}_S$ be the number of all outgoing edges from the whole set $S$ in $G_i^{\ell}$. It holds that $\mathbb{O}_S := \sum_{v_j \in S} \mathbb{O}_j$.
    Recall that the definition of $\Phi_{G_i^\ell}$ is the ratio of edges leading out of $S$ and all edges with at least one endpoint in $S$.
    Given that $G_i^{\ell}$ is a $\Delta^{\ell}$-regular graph, a simple calculation yields:
    \begin{align}
        \mathbb{E}\left[ \sum_{j=1}^s\sum_{k=1}^\frac{\Delta}{8} Y_j^k \right] &= \sum_{i=1}^s\sum_{j=1}^\frac{\Delta}{8}\mathbb{E}\left[ Y_j^k \right] 
        = \frac{\Delta}{8} \frac{\sum_{i=1}^s\mathbb{O}_{i}}{\Delta^{\ell}} = \frac{\Delta}{8}\frac{\mathbb{O}_S}{\Delta^{\ell}}\\
        = \frac{\Delta s}{8}  \Phi_{G_i^\ell}(S)
    \end{align}
    This proves the lemma.
    \end{proof}
Given this lemma, we can now prove the expected conductance. As we only observe regular graphs and w.h.p. all tokens create an edge, the expected conductance of a set $S$ simply follows from dividing $\mathcal{Y}_S$ by $\Delta s$, which immediately yields the lemma.
 Therefore, we get:
    \begin{lemma}
    Given that each node starts $\frac{\Delta}{8}$ tokens, which all create an edge, it holds:
    $$
        \mathbb{E}\left[ \Phi_{i+1}(S) \right] \geq  \mathbb{E}\left[ \frac{\mathcal{Y}}{\Delta s} \right] = \frac{\Phi_\ell(S)}{8}
    $$
    \end{lemma}
\begin{proof}
    Recall each nodes starts $\frac{\Delta}{8}$ tokens and we assume all tokens create an edge.
    By observing the algorithm, we note that by construction the degree can \emph{never} be higher than $\Delta$.
    Recall that every node creates its edges for $G_{i+1}$ based on the tokens it received.
    If any node receives fewer than $\Delta$ tokens, it creates self-loops to reach a degree of $\Delta$. 
    If it receives more, excess edges are dropped arbitrarily to ensure a degree of at most $\Delta$.
    Thus, each set $S$ maintains $\Delta|S|$ edges in total as each node will always has $\Delta$ edges irregardless of how many tokens is received.
    Given this fact and using Lemma \ref{lemma:exp_conductance} above we get 
    \begin{equation}
        \mathbb{E}\left[\frac{\sum_{i=1}^s\sum_{j=1}^d Y_i^j}{\Delta s} \right] = \frac{s\Delta}{8} \frac{\Phi^{\ell}}{s\Delta} = \frac{\Phi}{8}
    \end{equation}
\end{proof}

Therefore, a lower bound on $\Phi_{G_i^\ell}$ gives us a lower bound on the expected conductance of the newly sampled graph.
However, the standard Cheeger inequality (see, e.g., \cite{Sin12} for an overview) that is most commonly used to bound a graph's conductance with the help of the graph's eigenvalues does not help us in deriving a meaningful lower bound for $\Phi_{G_i^\ell}$.
In particular, it only states that $\Phi_{G_i^\ell} = \Theta(\ell\Phi_{G_i}^{2})$.
Thus, it only provides a useful bound if $\ell = \Omega(\Phi_{G_i}^{-1})$, which is too big for our purposes, as $\Omega(\Phi_{G_i}^{-1})$ is only constant if $\Phi_{G_i}$ is already constant. 
More recent Cheeger inequalities shown in \cite{LGT11} relate the conductance of smaller subsets to higher eigenvalues of the random walk matrix.
At first glance, this seems to be helpful, as one could use these to show that at least the small sets start to be more densely connected and then, inductively, continue the argument.
Still, even with this approach, constant length walks are out of the question as the new Cheeger inequalities introduce an additional tight $O(\log n)$ factor in the approximation for these small sets.
Thus, the random walks would need to be of length $\Omega(\log n)$, which is still too much to achieve our bounds.
Instead, we use the following result by Kwok and Lau~\cite{kwok2014lower}, which states that $\Phi_{G_i^\ell}$ improves even for constant values of $\ell$. 
It holds that:
\begin{lemma}[Conductance of $G^\ell$, Based on Theorem $1$ in \cite{kwok2014lower}]
\label{lemma:kwok_1}
    Let $G = (V,E)$ be any connected $\Delta$-regular lazy graph with conductance $\Phi_{G}$ and let $G^\ell$ be its $\ell$-walk graph.
    For a set $S \subset G$ define $\Phi_{G^\ell}(S)$ as the conductance of $S$ in $G^\ell$.
    Then, it holds:
    $$
        \frac{1}{2} \geq \Phi_{G^\ell}(S) \geq \max\left\{\frac{1}{40}\sqrt{\ell}\Phi_{G}, \, \Phi_{G}(S)\right\}  
    $$
\end{lemma}
Given this bound, we can show that benign graphs indeed increase their (expected) conductance from iteration to iteration.
In the following, we provide a sketch of the proof in \cite{kwok2014lower}.
Before we go into the details, we need another batch of definitions from the study of random walks and Markov chains.
Let $G:=(V,E)$ be s $\Delta$-regular, lazy graph and let $A_G \in \mathbb{R}^{n \times n}$ the stochastic random walk matrix of $G$.
Each entry $A_G(v,w)$ in the matrix has the value $\frac{e(v,w)}{\Delta}$ where $e(v,w)$ denotes the number of edges between $v$ and $w$ (or self-loops if $v=w$).
Likewise $A^\ell_G$ is the random walk matrix of $G^\ell$ where each entry has value $\frac{P_\ell(v,w)}{\Delta^\ell}$.
Note that both $A_G$ and $A^\ell_G$ are doubly-stochastic, which both their rows and their columns sum up $1$. 
For these types of weighted matrices, Kwok and Lau define the  \emph{expansion} $\varphi(S)$ of a subset $S \subset V$ as follows:
\begin{align}
    \varphi(S) = \frac{1}{|S|}\sum_{v \in S, w \in \overline{S}} A_G(v,w)
\end{align}
For regular graph (and \emph{only} those), this value is equal to the conductance $\Phi_{G}(S)$ of $S$, which we observed before.
This claim can be verified by the following elementary calculation:
\begin{align}
    \varphi(S) &= \frac{1}{|S|}\sum_{v \in S, w \in \overline{S}} A_G(v,w) = \frac{1}{|S|}\sum_{v \in S, w \in \overline{S}} \frac{e(v,w)}{\Delta}\\
    &= \frac{\sum_{v \in S, w \in \overline{S}} e(v,w)}{\Delta|S|}  =: \Phi_{G}(S) 
\end{align}
Therefore, the claim that Kwok and Lau make for the expansion also hold for the conductance of regular graphs\footnote{Indeed, they explicitly mention that for non-regular graph one could define a \emph{escape probability} for which their claims would hold and which could be used instead of the conductance in our proofs.
Nevertheless, since we only observe regular graphs, we use the notion of conductance to avoid introducing more concepts.}.
The proof in \cite{kwok2014lower} is based on the function $C^{(\ell)}(|S|)$ introduced by Lovász and Simonovits\cite{LS90}.
Consider a set $S \subset V$, then Lovasz and Simonovits define the following curve that bounds the distribution of random walk probabilities for the nodes of $S$. 
\begin{align}
    C^{(\ell)}(|S|) = \max_{\delta_0 + \dots + \delta_n = x, 0 \leq \delta_i \leq 1} \sum_{i=1}^n \delta_i (A^{\ell} p_S)_i
\end{align}
Here, the vector $p_S$ is the so-called characteristic vector of $S$ with $p_i = \frac{1}{|S|}$ for each $v_i \in S$ and $0$ otherwise.
Further, the term $(A^{\ell} p)_i$ denotes the $i^{th}$ value of the vector $A^{\ell} p_S$. 
Lovász and Simonovits used this curve to analyze the mixing time of Markov chains.
Kwok and Lau now noticed that it also holds that:

\begin{lemma}[Lemma 6 in \cite{kwok2014lower}]
\label{eqn:phi_lower_c}
It holds:
\begin{align}
    \Phi_{G^\ell}(S) \leq 1 - C^{(\ell)}(|S|)
\end{align}
\end{lemma}

Based on this observation, they deduce that a bound for $1-C^{(\ell)}(S)$ doubles as a bound for $\Phi_{G^\ell}$.
In particular, they can show the following bounds for $C^{(\ell)}(|S|)$:
\begin{lemma}[Lemma 7 in \cite{kwok2014lower}] 
It holds
\begin{equation}
    C^{(\ell)}(|S|) \leq 1-\frac{1}{20}\left(1-(1-\Phi_{G})^{\sqrt{\ell}}\right)
\end{equation}
\end{lemma}
Plugging these two insights together, we get
\begin{equation}
    \Phi_{G^\ell}(S) \geq 1 - C^{(\ell)}(|S|) \geq  \frac{1}{20}\left(1-(1-\Phi_{G})^{\sqrt{\ell}}\right) \geq \frac{\sqrt{\ell}}{40}\Phi_{G}
\end{equation}
The last inequality follows from the fact that $\sqrt{\ell}\Phi_{G}$ is at most $\frac{1}{2}$ and a standard approximation.
This is the main result of \cite{kwok2014lower}. 
We refer the interested reader to Lemma 7 of \cite{kwok2014lower} for the full proof with all necessary details. 
Their main technical argument is based on the following recusive relation between $C^{(\ell+1)}$ and $C^{(\ell)}$, which was (in part) already shown in \cite{LS90}: 
    \begin{lemma}[Lemma 1.4 in \cite{LS90}]
        \label{eqn:recursive}
        It holds
        $$ 
        C^{(\ell+1)}(|S|) \leq \frac{1}{2} \left(C^{(\ell)}(|S|+2\Phi_{G}\hat{|S|}) + C^{(\ell)}(|S|-2\Phi_{G}\hat{|S|})\right) 
        $$  
    \end{lemma}
    Here, we use the abbreviation $\hat{|S|} := \max\{|S|, n-|S|\}$.
    They use this to create a recursive formula that can be simplified to the given result using careful calculations. 
However, this fact alone is enough for our proof as the lower bound is too loose for subsets that already have a good conductance.
Instead we require that $\Phi_{G^\ell}(S)$ is \emph{at least as} big as $\Phi_{G}(S)$.
Note that this is not necessarily the case for all graphs. 
Instead, we must use the fact that our graphs are \emph{lazy}.
We show this in the following lemma:
 \begin{lemma}
    Let $G := (V,E)$ be any connected $\Delta$-regular lazy graph with conductance $\Phi_{G}$ and let $G^\ell$ be its $\ell$-walk graph.
    For a set $S \subset G$ define $\Phi_{G^\ell}(S)$ the conductance of $S$ in $G^\ell$.
    Then, it holds:
    $$
        \Phi_{G^\ell}(S) \geq \Phi_{G}(S)\
    $$
   \end{lemma}
\begin{proof}
Our proof is based on two claims.
First, we claim that $C^{(\ell)}(|S|)$ is monotonically increasing in $\ell$.
\begin{claim}
    It holds $C^{(\ell)}(|S|) \leq C^{(\ell-1)}(|S|)$
\end{claim}
\begin{proof}
This fact was already remarked in \cite{LS90} based on an alternative formulation. However, given that $C^{(\ell)}$ is concave, it holds that for all values $\gamma, \beta \geq 0$ with  $\gamma \leq \beta$ that
\begin{equation}
        \label{eqn:concave}
        C^{(\ell)}(S+\beta\hat{S}) + C^{(\ell)}(|S|-\beta\hat{S}) \leq C^{(\ell)}(S+\gamma\hat{|S|}) + C^{(\ell)}(|S|-\gamma\hat{|S|})   
    \end{equation}
    And thus, together with \autoref{eqn:recursive}, we get:
    \begin{align}
        C^{(\ell)}(|S|) &\leq \frac{1}{2} \left(C^{(\ell-1)}(|S|+2\Phi_{G}|\hat{S}|) + C^{(\ell-1)}(|S|-2\Phi_{G}|\hat{S}|)\right)\\
        &\leq \frac{1}{2} \left(C^{(\ell-1)}(|S|+0\cdot|\hat{S}|) + C^{(\ell-1)}(|S|-0\cdot|\hat{S}|)\right) \\
        &= C^{(\ell-1)}(|S|)
    \end{align}
    Here, we chose $\beta=2\Phi_{G}$ and $\gamma=0$ and applied Equation \ref{eqn:concave}. This proves the first claim.
    \end{proof}
Second, we claim that $C^{(1))}(|S|)$ is \emph{equal} to $1-\Phi_{G}(S)$ as long as the graph we observe is lazy.
\begin{claim}
    It holds $C^{(1)}(|S|) = 1-\Phi_{G}(S)$
\end{claim}
\begin{proof}
For this claim (which was not explicitly shown in \cite{kwok2014lower}, but implied in \cite{LS90}) we observe 
\begin{align}
C^{(1)}(S) = \max_{\delta_0 + \dots + \delta_n = x, 0 \leq \delta_i \leq 1} \sum_{i=1}^n \delta_i (A_G p_S)_i    
\end{align}
 and find the assignment of the $\delta$'s that maximizes the sum.
     Lovasz and Simonovits already remarked that it is maximized by setting $\delta_i=1$ for all $v_i \in S$.
     However, since there is no explicit lemma or proof to point to in \cite{LS90}, we prove it here.
     First, we show that all entries $(A_Gp_S)_i$ for nodes $v_i \in S$ are least $\frac{1}{2|S|}$ and all entries $(A_Gp_S)_{i'}$ for nodes $v_{i'} \not\in S$ are at most $\frac{1}{2|S|}$.  
     We begin with the nodes in $S$.
     Given that $G$ is $\Delta$-regular and lazy, we have for all $v_i \in S$ that  
     \begin{align}
         (A_Gp_S)_i &= \sum_{j=1}^n A_G(v_i,v_j) {p_S}_j \geq A_G(v_i,v_i) {p_S}_i \geq \frac{1}{2|S|}.  
     \end{align}
     Here, ${p_S}_i = \frac{1}{|S|}$ follows because $v_i \in S$ per definition.
     The inequality $A_G(v_i,v_i) \geq \frac{1}{2}$ follows from the fact that $A$ is lazy and each node has a self-loop with probability $\frac{1}{2}$.
     As a result, the entry $(A_Gp_S)_i$ for $v_i \in S$ has at least a value of $\frac{1}{2|S|}$, even if it has no neighbors in $S$.
     On the other hand, we have for all nodes $v_{i'} \not\in S$ that
     \begin{align}
         (A_Gp_S)_{i'} &= \sum_{j=1}^n A_G(v_j, v_{i'}) p_j = \sum_{v_j \in S} A_G(v_{i'},v_j)  \frac{1}{|S|}
     \end{align}
     This follows from excluding all entries $p_j$ with $v_j \not\in S$. Note that for these values it holds $p_j=0$.
     Further, Since $A$ is $\Delta$-regular and lazy, each node $v_{i'} \not\in S$ has at most $\frac{\Delta}{2}$ edges to nodes in $S$.  
     \begin{align}
         (A_Gp_S)_{i'} = \sum_{v_j \in S} A_G(v_i,v_j)  \frac{1}{|S|}
          \leq \frac{\Delta}{2}\frac{1}{\Delta}\frac{1}{|S|} = \frac{1}{2|S|} 
     \end{align}
     Thus, the corresponding value $(A_Gp_S)_i$ of any $v_i \in S$ is \emph{at least} as big as value $(A_Gp_S)_{i'}$ of $v_{i'} \not\in S$. 
     By a simple greedy argument, we now see that $\sum_{i=1}^n \delta_i (A^{\ell} p_S)_i$ is maximized by picking $\delta_i = 1$ for all nodes in $S$: 
     To illustrate this, suppose that there is a choice of the $\delta$'s such that $\sum_{i=1}^n \delta_i (A_Gp_s)_i$ is maximized and it holds $\delta_i < 1$ for some $v_i \in S$.
     Since no $\delta$ can be bigger than $1$ and the $\sum_{i = 1}^n \delta_i = |S|$ there must be a $v_{i'} \not\in S$ with $\delta_{i'}>0$.
     Since $(A_Gp_S)_i \geq (A_Gp_S)_{i'}$ decreasing $\delta_{i'}$ and increasing $\delta_{i}$ does not decrease the sum.
     Thus, choosing $\delta_i=1$ for all $v_i \in S$ must maximize the term $\sum_{i=1}^n \delta_i (A_Gp_s)_i$. 
     Thiy yields:
     \begin{align}
         \sum_{i=1}^n \delta_i (A_Gp_S)_i &= \sum_{v_i \in S}\sum_{v_j \in S} A_G(v_i,v_j) \frac{1}{|S|} = \frac{1}{|S|} \sum_{v_i \in S} \frac{e(v_i,v_j)}{\Delta}\\
         &= \frac{\Delta|S|-\mathbb{O_S}}{\Delta|S|} = 1 - \frac{\mathbb{O}_S}{\Delta|S|} = 1-\Phi_{G}(S)
     \end{align}
     Here, the value $\mathbb{O_S}$ denotes the edges leaving $S$. 
     Given that the graph is $\Delta$-regular, the term $\Delta|S|-\mathbb{O_S}$ counts all edges in $S$.
     This was to be shown.
     \end{proof}
If we combine our two claims, the lemma follows.
\end{proof}
Thus, as long as $G_i$ is regular and lazy, we have a suitable lower bound for $\Phi_{G_i^\ell}$.
In fact, we can show the following:
\begin{lemma}
\label{lemma:chernoff_outgoing}
    Let $S \subset V$ be set of nodes with $O_S$ outgoing edges, then it holds:
    \begin{align}
        \pr{\Phi_{G_{i+1}}(S) \leq \frac{\sqrt{\ell}\Phi_{G_i}}{640}} \leq e^{-\frac{O_S}{64}}
    \end{align}
\end{lemma}
This follows from the Chernoff bound and the fact that the random walks are quasi independent
\footnote{Technically, they are not independent since in the last step we draw from all tokens without replacement.
However, since we condition on all nodes receiving less than $\frac{3\Delta}{8}$, we can observe the experiment where each node start $\frac{\Delta}{8}$ independent random walks and connects with the endpoints. Here, the Chernoff bound holds.}.
However, this fact alone is not enough to finalize the proof of Lemma \ref{lemma:induction_step}.
Recall that we need to show that \emph{every} subset has a conductance of $O(\sqrt{\ell}\Phi_{G_i})$ in $G_{i+1}$ in order to prove that $\Phi_{G_{i+1}} = \Omega(\sqrt{\ell}\Phi_{G_i})$. 
Since there are exponentially many subsets, a bound in the magnitude of $o(n^{-c})$ is not sufficient on its own.
Luckily, by a celebrated result of Karger, the number of subsets with $\alpha\Lambda$ outgoing edges can be bounded by $O(n^{2\alpha})$ \cite{karger2000minimum}, i.e, it holds
\begin{theorem}[Theorem 3.3 in \cite{karger2000minimum}, simplified]
\label{lemma:karger}
    Let $G$ be an undirected, unweighted graph and let $\Lambda>1$ be the size of a minimum cut in $G$.
    For an even parameter $\alpha \geq 2$ the number of cuts with at most $\alpha \Lambda$ edges is bounded by $2\binom{n}{2\alpha}$.
\end{theorem}
Given that a set $S$ has $\alpha\Lambda$ outgoing edges, the value of $\Phi_{G_{i+1}}(S)$ is concentrated around its expectation with probability $e^{-\Omega{(\alpha\Lambda)}}$.
Thus, for a big enough $\Lambda$, a bound of $e^{-\Omega{(\alpha\Lambda)}}$ is enough to show all sets increase their conductance. 
With this insight, we can prove Lemma \ref{lemma:induction_step}:
For a set $S \subset V$ we define $\mathcal{B}_S$ to be the event that $\Phi_{G_{i+1}}(S)$ is smaller than $\frac{1}{640}\sqrt{\ell}\Phi_{G_i}$. 
In this case, we say that $S$ has \emph{bad conductance}. 
Obviously, if no set has bad conductance, then the resulting conductance of $G_{i+1}$ must also be at least $\frac{1}{640}\sqrt{\ell}\Phi_{G_i}$ and lemma follows.
We let $\mathcal{B} = \bigcup_{S \subset V} \mathcal{B}_S$ be the event that there exists a set $S$ with bad conductance, i.e, there is any $\mathcal{B}_S$ that is true.
To prove the lemma, we show that $\mathcal{B}$ does not happen w.h.p.
Therefore, we let $\mathcal{S}_\alpha \in \mathcal{P}(V)$ be the set of all sets that have a cut of size $c \in [\alpha\Lambda,2\alpha\Lambda)$.
Further, we let $\Lambda \geq 6400\lambda\log n$ for a constant $\lambda$.
Note that $\lambda$ can be chosen as high as we want by constructing a sufficiently large minimum cut in $G_0$ by creating copies of each initial edge.
Using all these definitions we can show that the following holds:
    \begin{align}
        \pr{\mathcal{B}} 
        &\leq\sum_{S \subset V} \pr{\mathcal{B}_S }\leq\sum_{\alpha = 1}^\frac{\Delta n}{\Lambda} \sum_{S \in \mathcal{S}_\alpha} \pr{\mathcal{B}_S \,\big|\, S \in \mathcal{S}_\alpha } \\
        &\leq\sum_{\alpha = 1}^\frac{\Delta n}{\Lambda} \sum_{S \subset V, cut(S,\overline{S}) \leq 2\alpha\Lambda} \pr{\mathcal{B}_S \,\big|\, S \in \mathcal{S}_\alpha }
        \end{align}
        Now we can apply Theorem \ref{lemma:karger} and see that
        \begin{align}
        \pr{\mathcal{B}}  &\leq 2\sum_{\alpha = 1}^\frac{\Delta n}{\Lambda}\binom{n}{2 \cdot 2\alpha} \pr{\mathcal{B}_S \,\big|\, S \in \mathcal{S}_\alpha} \\
        &\leq 2\sum_{\alpha = 1}^\frac{\Delta n}{\Lambda}\binom{n}{4\alpha} \pr{\Phi_{G_{i+1}}(S) \leq \frac{1}{640}\sqrt{\ell}\Phi_{G_i}  \,\big|\, \mathbb{O}_S \geq \alpha\Lambda} 
        \end{align}
        By using the Lemma \ref{lemma:chernoff_outgoing}, we get:
        \begin{align}
        \pr{\mathcal{B}}  &\leq 2\sum_{\alpha = 1}^\frac{\Delta n}{\Lambda}\binom{n}{4\alpha} e^{-\frac{\Lambda\alpha}{64}} 
        \leq\sum_{\alpha = 1}^\frac{\Delta n}{\Lambda}\binom{n}{4\alpha} n^{-10\cdot\lambda\alpha} \\
        &\leq 2\sum_{\alpha = 1}^\frac{\Delta n}{\Lambda} \left(\frac{en}{4\alpha}\right)^{4\alpha} n^{-10\lambda\alpha} 
     \leq\frac{\Delta n}{\Lambda} n^{-5\lambda} = n^{-4\lambda}
    \end{align}
Thus, Lemma \ref{lemma:induction_step} follows for any constant $\lambda$ hidden in $\Lambda = \Omega(\log n)$.

\subsection{Ensuring That Each \text{$G_i$} is Benign}
We will show that indeed each $G_i$ is a $\Delta$-regular, lazy graph with a $\Lambda$-sized cut. Note that the last property also ensures that $G_i$ is connected. While the first property follows directly from the algorithm, the latter two require some closer observations.
 
\paragraph*{Ensuring that $G_i$ is $\Delta$-regular.}
By observing the algorithm, we note that by construction the degree can \emph{never} be higher than $\Delta$.
Recall that every node creates its edges for $G_{i+1}$ based on the tokens it received.
If any node receives fewer than $\Delta$ tokens, it creates self-loops to reach a degree of $\Delta$. 
If it receives more, excess edges are dropped arbitrarily to ensure a degree of at most $\Delta$.

\paragraph*{Ensuring that $G_i$ is lazy.}
For this, recall that a node connects to endpoints of all its $\frac{\Delta}{8}$ tokens and additionally to the origins of all (but at most $\frac{3\Delta}{8}$) tokens it received.
Thus, in the worst case, it creates $\frac{\Delta}{8}+\frac{3\Delta}{8} = \frac{\Delta}{2}$ outgoing edges.
Thus, it creates at least $\frac{\Delta}{2}$ --- and therefore enough --- self-loops.

\paragraph*{Ensuring that $G_i$ has $\Lambda$-sized minimum cut.}
The third property, the $\Lambda$-sized minimum cut, is perhaps the most difficult to show.
However, at a closer look, the proof is almost identical to the proof of Lemma \ref{lemma:induction_step}. 
In particular, we show that all cuts that are \emph{close} to the minimum cut will (in expectation and w.h.p.) increase in their size in each iteration, but never fall below $\Lambda$.
The idea behind the proof uses the fact that \cite{kwok2014lower} actually gives us a stronger bound on the expected growth of the subset than just the conductance.
In fact, for each subset $S$, it suffices to observe subsets of similar size to get a lower bound for $\Phi_\ell(S)$.
Before we go into more details, we recall the notion of small-set conductance, which is a natural generalization of conductance:
\begin{definition}[Small-Set Conductance]
    Let $G := (V,E)$ be a connected $\Delta$-regular graph and $S \subset V$ with $|S| \leq \frac{\delta|V|}{2}$ be any subset of $G$.
    The small-set conductance $\Phi_{\delta}$ of $G$ is then defined as:
    $$
        \Phi_\delta(G) := \min_{S \subset V, |S| \leq \frac{\delta|V|}{2}} \Phi(S) 
    $$
\end{definition}
Given this definition, we note that there is also a (weaker) bound on the small set conductance of $G^\ell$ in \cite{kwok2014lower}. It holds:
\begin{lemma}[Small-Set Conductance of $G^\ell$, Theorem $3$ in \cite{kwok2014lower}]
\label{lemma:small_set_kwok}
    Let $G := (V,E)$ be any connected $\Delta$-regular lazy graph with small-set conductance $\Phi_\delta$ for any $\delta \in (0,1) $and let $G^\ell$ be its $t^{\text{th}}$ power.
    For a set $S \subset G$ with $|S| \leq \frac{\delta n}{2}$ define $\Phi_\ell(S)$ the conductance of $S$ in $G^\ell$.
    Then, it holds:
    $$
        \frac{1}{4} \geq \Phi_\ell(S) \geq \max\{\frac{1}{40}\sqrt{\ell}\Phi_{\delta},\Phi(S)\}  
    $$
\end{lemma}
Therefore, we see that that sets $S \subset S$ of size $\frac{\delta \cdot n}{2}$ whose conductance $\Phi(S)$ is close to $\Phi_{\delta}$ have many tokens that end outside of $S$.

One can easily see that there is a simple relation between the small-set conductance and the minimum cut of a graph.
Since the conductance of a set is the number of its outgoing edges divided by its size, the minimum cut gives us a simple lower bound \emph{all} small set conductance for all values of $\delta$. It holds:
\begin{lemma}
\label{lemma:min_small_set}
    Let $G := (V,E)$ have minimum cut of $\Lambda$, then the small-set conductance $\Phi_\delta$ of $G$ is at least $\frac{\Lambda}{\Delta \delta n}$.
\end{lemma}
\begin{proof}
Suppose there is set $S \subset V$ and $|S| \leq \Delta\delta n$ with conductance smaller than $\frac{\Lambda}{\Delta \delta n}$. 
Then, either the number of outgoing edges must be smaller, or size of the set must be bigger.
However, since the number of outgoing edges of $S$ is at least $\Lambda$, it must hold that $|S| > \Delta \delta n$. 
This is a contradiction since we assumed that $|S| \leq \Delta\delta n$.
\end{proof}
This is simple observation this enough to show that all sets that have close to $\Lambda$ outgoing connections slightly increase the number for their outgoing connections for a big enough $\ell$.
In particular, it holds:
\begin{lemma}
\label{lemma:min_cut_prob}
    Suppose that $\ell > 2\cdot640^2$.
    Then, for any set $S$ with $O_S$ outgoing edges, it holds:
    $$
        \pr{\mathcal{Y}_S \leq \Lambda} \leq e^{-\frac{1}{8}\max\left\{2\Lambda,\frac{O_S}{4}\right\}}
    $$
\end{lemma}
\begin{proof}
The proof follows the same basic structure as before.
Recall that for each set $S$, the number of outgoing edges $\mathcal{Y}_S := \sum^{s}_{i=1} \sum_{k=1}^{\Delta/8} Y_i^k$ in $G_{i+1}$ is determined by a series (quasi-)independent binary variables.
Thus, by the Chernoff Bound, it holds that
\begin{equation}
     \pr{\mathcal{Y} \leq (1-\delta)\mathbb{E}[\mathcal{Y}]} \leq e^{\frac{\delta^2}{2}\mathbb{E}[\mathcal{Y}]}.
\end{equation}
Now, we claim that it holds $\mathbb{E}[\mathcal{Y}] \geq \max\{2\Lambda,\frac{O_S}{4}\}$ and choosing $\delta \geq \nicefrac{1}{2}$, we get
$$
    \pr{\mathcal{Y} \leq \Lambda} \leq e^{\frac{1}{8}\max\{2\Lambda,\frac{O_S}{4}\}}.
$$
Therefore, it remains to show that our claim that $\mathbb{E}[\mathcal{Y}] \geq \max\{2\Lambda,\frac{O_S}{4}\}$ holds true and we are done.
Now, we distinguish between two cases:
\begin{description}
    \item[Case 1: $\mathbb{O}_S \geq 8\Lambda$]
    By Lemma \ref{lemma:exp_conductance} we have for all set with $O_S$ outgoing edges that 
    \begin{align}
        \mathbb{E}[\mathcal{Y}_S] \geq \frac{O_S}{4}.
    \end{align}
    Thus, for $O_S > 8\Lambda$ we have $\mathbb{\mathcal{Y}_S} \geq \frac{O_S}{4}$ and the lemma follows.
    \item[Case 2: $\mathbb{O}_S < 8\Lambda$]
    In the following, we consider only sets $S \subset V$ with fewer than $8\Lambda$ outgoing edges. Consider a set of size $|S| = s = \delta_s n$.
By Lemma \ref{lemma:min_small_set} it holds that $\Phi_{2\delta_s} \geq \frac{\Lambda}{2\Delta s}$.
Thus, again using Lemma \ref{lemma:exp_conductance}, it holds
\begin{align}
    \mathbb{E}[\Phi(S)] &\geq \frac{\Phi_\ell}{8} & \rhd \autoref{lemma:exp_conductance} \\
    &\geq \frac{1}{8}\frac{1}{40}\sqrt{\ell}\Phi_{2\delta_s} & \rhd \autoref{lemma:small_set_kwok} \\
    &\geq \frac{1}{320}\sqrt{\ell}\frac{\Lambda}{2 \Delta s} & \rhd \autoref{lemma:min_small_set} \\
    &\geq \frac{1}{640}\sqrt{\ell}\frac{\Lambda}{\Delta s}.
\end{align} 
The factor of $2$ that appears in the denominator in the second line results from \autoref{lemma:small_set_kwok}.
Since we observe a set of size $\delta_s n$, we must consider $\Phi_{2\delta_s}$.
Since we always consider sets of size at most $\frac{n}{2}$, this is always well-defined.
Since the number of edges in $S$ stays constant, we can again conclude
\begin{align}
    &\mathbb{E}[\Phi(S)] \geq \frac{1}{640}\sqrt{\ell}\frac{\Lambda}{\Delta s}\\
    \Leftrightarrow &\mathbb{E}[\frac{\mathcal{Y}_S}{\Delta s}] \geq \frac{1}{640}\sqrt{\ell}\frac{\Lambda}{\Delta s}\\  
    \Leftrightarrow &\mathbb{E}[\mathcal{Y}_S] \geq \sqrt{\ell}\frac{\Lambda}{640}.
\end{align}
By choosing a sufficiently large $\ell > 2 \cdot 640^2$, we get $\mathbb{E}[\mathcal{Y}_S] > 2\Lambda$ as desired.
\end{description}

\end{proof}
We can round up the proof by the same trick as before.
Again, we must show that \emph{every} cut has a value of at $\Lambda$ and use Karger's bound together with Lemma \ref{lemma:min_cut_prob} to show that no cut has a worse value, w.h.p.

\begin{lemma}
If $G_i$ is begnin, the minimum cut of $G_{i+1}$ is at least $\Lambda$ w.h.p.
\end{lemma}
\begin{proof}
For a set $S \subset V$ we define $\mathcal{C}_S$ to be the event that $\mathcal{Y}_S$ is smaller than $\Lambda$. 
    In this case, we say that $S$ has \emph{bad cut}. 
    Obviously, if no set has bad cut, the minimum cut of $G_{i+1}$ is $\Lambda$
    We let $\mathcal{C} = \bigcup_{S \subset V} \mathcal{C}_S$ be the event that there exists a set $S$ with bad cut, i.e, there is any $\mathcal{C}_S$ that is true.
    To prove the lemma, we show that $\mathcal{C}$ does not happen w.h.p.
    As before, we let $\mathcal{S}_\alpha \in \mathcal{P}(V)$ be the set of all sets that have a cut of size $c \in [\alpha\Lambda,2\alpha\Lambda)$.
    Further, we let $\Lambda \geq \lambda\log n$ for a constant $\lambda$, i.e, we denote the constant hidden in the $O$-Notation as $\lambda$.
    Note that $\lambda$ can be chosen as high as we want by constructing a sufficiently large minimum cut in $G_0$ by creating copies of each initial edge.
    Using all these definitions we can show that the following holds:
      \begin{align*}
        \pr{\mathcal{C}} 
        &\leq\sum_{S \subset V} \pr{\mathcal{C}_S } & \rhd \textit{Union Bound}\\
        &\leq\sum_{\alpha = 1}^\frac{\Delta n}{\Lambda} \sum_{S \in \mathcal{S}_\alpha} \pr{\mathcal{C}_S \,\big|\, S \in \mathcal{S}_\alpha } & \rhd \textit{Regrouping}\\
        &\leq\sum_{\alpha = 1}^\frac{\Delta n}{\Lambda} \sum_{S \subset V, cut(S,\overline{S}) \leq 2\alpha\Lambda} \pr{\mathcal{C}_S \,\big|\, S \in \mathcal{S}_\alpha } & \rhd \textit{Definition of } S_{\alpha}\\
        &\leq\sum_{\alpha = 1}^\frac{\Delta n}{\Lambda}\binom{n}{2 \cdot 2\alpha} \pr{\mathcal{C}_S \,\big|\, S \in \mathcal{S}_\alpha} & \rhd \textit{\autoref{lemma:karger}}\\
        &\leq\sum_{\alpha = 1}^\frac{\Delta n}{\Lambda}\binom{n}{4\alpha} \pr{\Phi_{G_{i+1}} \leq \frac{1}{640}\sqrt{\ell}\Phi_i  \,\big|\, \mathbb{O}_S \geq \alpha\Lambda} & \rhd \textit{Definition of } \mathcal{C}_S\\
        &\leq\sum_{\alpha = 1}^\frac{\Delta n}{\Lambda}\binom{n}{4\alpha} e^{-\frac{\Lambda\alpha}{64}} & \rhd \textit{\autoref{lemma:min_cut_prob}}\\
        &\leq\sum_{\alpha = 1}^\frac{\Delta n}{\Lambda}\binom{n}{4\alpha} n^{-10\cdot\lambda\alpha} & \rhd \textit{Using that } \Lambda\geq 640 \cdot \lambda \cdot \log n\\
        &\leq\sum_{\alpha = 1}^\frac{\Delta n}{\Lambda} \left(\frac{en}{4\alpha}\right)^{4\alpha} n^{-10\lambda\alpha} 
        &\rhd \textit{Using that } \binom{n}{k} \leq \left(\frac{en}{k}\right)^k\\
        &\leq\frac{\Delta n}{\Lambda} n^{-5\lambda} = n^{-4\lambda}\\
    \end{align*}
    Thus, the lemma follows for a sufficently large constant $\lambda$ hidden in $\Lambda = \Omega(\log n)$.
\end{proof}

\subsection{Finalizing the Proof}

To round up the analysis, we now only need to show that after $O(\log n)$ iterations, the graph has constant conductance.
Based on our insights, we can conclude that if $\Delta,\ell$ and $\Lambda$ are big enough, then, w.h.p., if $G_i$ is benign, then $G_{i+1}$ is benign and has at least twice its conductance (if it was not already constant).
In particular, the following three events hold true w.h.p.
\begin{enumerate}
    \item For $\Delta > 8 k \log(n)$ it holds with probability $1-\frac{\ell}{n^k}$ that all node receive less then $\frac{3\Delta}{8}$ token each round. We call this event $\mathcal{E}_1$.
    \item For $\ell > 2\cdot 640^2$ and $\Lambda > 640 k \log(n)$ consider the experiment that every node picks $\frac{\Delta}{8}$ nodes through independent $\ell$-step random walks in $G_i$.
    Then, by Lemma \ref{lemma:induction_step}, the resulting graph $G_{i+1}$ has conductance at least $2\Phi$ (if $\Phi$ was not already constant).
    We call this event $\mathcal{E}_2$.
    \item For $\ell > 2\cdot 640^2$ and $\Lambda > 640 k \log(n)$ the minimum cut of $G_{i+1}$ is again at least $\Lambda$. We call this event $\mathcal{E}_3$. 
\end{enumerate}
Note that given $\mathcal{E}_1$, the algorithm can modeled as the experiment described above.
By the union bound, the events $\mathcal{E}_1,\mathcal{E}_2$ and $\mathcal{E}_3$ hold together w.h.p.

Given that after every iteration, the graph's conductance increases by a factor $O(\sqrt{\ell})$ w.h.p, a simple union bound tells us, as long as we consider $o(n)$ iterations, the conductance increases in every iteration w.h.p.
Thus, after $O\left(\frac{\log \Phi}{\log\ell}\right) = O(\log n)$ iterations, the most recent graph must have constant conductance, w.h.p. 
Therefore, since each iteration lasts only $\ell = O(1)$ rounds, after $O(\log n)$ rounds, the graph has a constant conductance and logarithmic diameter.
This concludes the analysis.


\clearpage
\section{Applications in the Hybrid Model} \label{sec:applications}

We now present some applications of our algorithm in the hybrid model.
Note that in this section, we will use the fact that a node can communicate with \emph{all} of its neighbors in $G$ via small messages, i.e, we assume the \CONGEST model for $G$.
Recall that this is a necessary assumption to achieve a runtime that is independent of the graph $G$'s degree (or arboricity). 
However, the global capacity, which bounds the total number of messages a node can send and receive via global edges, is bounded by $\widetilde{O}(1)$.
Before we approach the different graph problems in the hybrid model, we first give an adaption of \textsc{CreateExpander} to this model that circumvents some problems introduced by the potentially high node degrees.
This algorithm will be the basis of all algorithms in the remainder of this section.

\subsection{Adapting {\normalfont\scshape{CreateExpander}} to the Hybrid Model}
\label{sec:adaption}

In Section \ref{sec:algorithm} we used a very simple approach to construct the initial benign graph on which \textsc{CreateExpander} is executed: copy each initial edge $O(\log n)$ times.
While this technique works for $O(1)$-regular graphs, extending it to more general input graphs introduces some difficulties.
In particular, for graphs of degree $d = \omega(1)$ the benign graph's degree becomes $\Delta = \omega(\log n)$.
Although the algorithm itself could handle this case if we increased the allowed communication capacity to $O(\Delta)$\footnote{Note that this is still polylogarithmic, if $O(d)$ is polylogarithmic.}, the resulting overlay network would have superlogarithmic degree, violating the definition of a well-formed tree.
Additionally, this prevents us from using our main algorithm as a black box for the applications in the following section.
In this section, we approach this problem by presenting a variant of Theorem~\ref{thm:main} for the hybrid model that allows an initial degree of $d = O(\log n)$.
In particular, we will use slightly longer random walks (of logarithmic length) and more communication per node.
However, since each node only communicates a polylogarithmic number of messages over global edges, the algorithm directly works in the models of \cite{GHS19, GHSS17}.
The main contribution of this section is the following adapted version of Theorem \ref{thm:main}:
\begin{theorem}\label{thm:main_variant}
    Let $G = (V,E)$ be a weakly connected directed graph with degree $d = O(\log n)$ and $m$ nodes.
    There is a randomized algorithm that constructs a well-formed tree $T_G = (V,T_V)$ in $O(\log m+\log\log n))$ rounds, w.h.p., in the hybrid model.
    The algorithm requires global capacity $O(\log^3 n)$. 
\end{theorem}
At the core of our adapted algorithm lies the following technical theorem, which was independently shown by \cite{DGS16}, \cite{AS18}, and \cite{LMOS20}.
It assures us that we can simulate random walks of length $\ell$ in time $O(\log\ell)$ in overlay networks, given that we have sufficient communication bandwidth.
\begin{lemma}[Rapid Sampling, \cite{DGS16,AS18,LMOS20}]
\label{lemma:rapid_sampling}
    Let $G$ be a $d$-regular graph.
    If each node can send and receive $O(m \ell)$ messages of size $O(\log n)$ in each round, then each node can sample $m = \Omega(\log n)$ random walks of length $\ell$ in time $O(\log\ell)$, w.h.p.
\end{lemma}
The main idea behind the algorithm is to \emph{stitch} short random walks to longer ones while maintaining their independence.
For the first $2$ rounds, all random walk token are forwarded as usual, i.e., for each token an incident edge is picked uniformly at random.
Then, after these $2$ rounds, we stitch the random walks together to quickly double their length (at the cost of reducing their number by half).
In particular, in each round, each node selects half of the random walk tokens it received and marks them as red. 
All other token are marked as blue.
If the node has an odd number of token, the remaining token is simply dropped.
For each red token, every node picks one of the remaining blue tokens uniformly at random and without replacement. 
Then, it sends the red token to the blue token's origin.
Afterwards, the blue token is discarded to maintain the independence of the red walks.
One can easily verify the following three facts:
\begin{enumerate}
    \item After $\log\ell$ rounds, all surviving tokens are distributed according to random walks of length $\ell$. 
    This follows because combining a red and blue token effectively doubles the length of the walk. 
    Also, since each each token has an equal probability of becoming red, there is no significant bias in the distribution. 
    \item All surviving random walks are independent. 
    This follows because we never reuse a blue token and thus, the surviving walks are uncorrelated.
    \item If $\Delta=O(\log n)$ and each node starts $O(\Delta \ell)$ tokens, then $O(\Delta)$ of them survive, w.h.p. 
    This follows because every token that never becomes blue survives.
    Since a token becomes blue with probability $\nicefrac{1}{2}$ in each round and there are $\log\ell$ rounds, in expectation $O(\frac{\Delta\ell}{{\ell}}) = O(\Delta)$ survive. 
    The rest follows from Chernoff bounds\footnote{Note that Chernoff also applies to drawing without replacement.} and sufficiently large $\Delta$. 
\end{enumerate}
For a more detailed analysis, we refer to \cite{DGS16,AS18,LMOS20}, where all of these claims are proved in detail.

We now adapt three implementation details of \textsc{CreateExpander}.
First, instead of initially copying each edge $O(\Lambda)$ times, each node only adds self-loops until it reaches degree $\Delta > 2d$.
Second, in each evolution, we use the rapid sampling technique instead of normal random walks.
Next, we need to consider that --- if we use rapid sampling --- a node cannot control how many of its tokens succeed.
Therefore, all tokens that survived for $\ell$ rounds are sent back to their origin together with the identifier of their endpoint.
Then, each node picks $\frac{\Delta}{8}$ of these tokens to create edges and answers back to the respective endpoints.
Last, we choose $\ell = O(\Lambda^2)$.
We will show that a) this causes the minimum cut and the conductance to grow by $O(\Lambda)$ in expectation and w.h.p., and b) the runtime is $O(\log m+\log\log n)$ w.h.p.

For the first statement, we observe that a large portion of the analysis in Section~\ref{sec:analysis} remains valid for the adapted algorithm.
Note that the creation of edges is changed as we added an extra step to determine the tokens that survived.
This slightly changes the random experiment that bounds the number of outgoing edges of each set.
Since Lemma \ref{lemma:rapid_sampling} guarantees us that, w.h.p., enough tokens survive and their distribution is unchanged, the expected value in Lemma \ref{lemma:exp_conductance} stays the same.
Further, the Chernoff bound still applies because the random walks stay independent.
Thus, given that we have big enough minimum cut, the algorithm works as before.

Now observe our analysis of the minium cut, especially Lemma \ref{lemma:min_cut_prob}, we can easily see that the following holds:
If we choose $\ell = O(\Lambda^2)$, then Lemma \ref{lemma:min_cut_prob} holds regardless of the initial cut's size for a large enough $\ell$ for the Chernoff Bound to kick in. 
Thus, after the first evolution, the minimum cut is of size $\Lambda$, w.h.p. and the increase in conductance is $\Theta(\Lambda)$ as well.
Therefore, the algorithm works, w.h.p. without copying each edge $O(\Delta)$ times.

Finally, we must observe how these longer walks affect the runtime.
Since $\Lambda^2 = O(\log^2 n)$, we can simulate the random walks in $O(\log\log n)$ rounds instead of using normal walks.
This requires a global capacity of $O(\log^3 n)$ by Lemma \ref{lemma:rapid_sampling}.
Thus, a single evolution takes $O(\log\log n)$ rounds instead of $O(1)$.
However, due to the longer walks, also have an increase of $\Theta(\log n)$ in the conductance.
This follows directly from Lemma \ref{lemma:induction_step}.
Now observe that the minimal conductance of any graph with $m$ nodes and degree $\Delta$ is $\frac{1}{\Delta m}$
Therefore, we only need $L' := O(\frac{\log m}{\log\log n})$ evolutions to obtain an expander and the overall runtime is $O(\log m + \log\log n)$.

\subsection{Connected Components}
\label{sec:components}

In this section, we show how the algorithm can be used to find connected components in an arbitrary graph $G$.
In particular, for each connected component $C$ of $G$, want to establish a well-formed tree (overlay edges) that contains all nodes of $C$.
The main result of this section is the following theorem:

\connected*

Note that the main difficulty that prevents us from applying Theorem~\ref{thm:main_variant} is the fact that the initial graph's degree is unbounded.
Therefore, the main contribution of this section is an algorithm that transforms any connected subgraph of $G$ into a graph $H$ of bounded degree $O(\log n)$.
While this can be achieved using spanner constructions in time $O(\log n)$, a more careful analysis is required to show that we can construct such a graph in time $O(\log m + \log \log n)$.
After constructing $H$, we execute the algorithm of Theorem~\ref{thm:main_variant} to create a well-formed tree for each component.
By Theorem~\ref{thm:main_variant}, this takes time $O(\log m + \log \log n)$, w.h.p.
Therefore, we only need to prove the following lemma.

\begin{lemma} \label{lem:connected_preprocessing}
    Let $G = (V,E)$ be a directed graph in which each component contains at most $m$ nodes.
    There exists a randomized algorithm that transforms $G$ into a directed graph $H := (V, E_H)$ that has degree $O(\log n)$ and in which two nodes lie in the same component if and only if they lie in the same component in $G$.
    The algorithm takes $O(\log m + \log \log n)$ rounds w.h.p., in the \CONGEST model.
    %
    %
\end{lemma}

The algorithm's main idea is to first eliminate \emph{most} edges by constructing a sparse spanner.
Then, in a second step, we let all remaining nodes of high degree \emph{delegate} their edges to nodes of lower degree.
If every node of high degree has sufficiently many neighbors of low degree, the overall degree becomes small enough for our algorithm to handle.
Here, we use the fact that most modern spanner construction algorithms create spanners with exactly this property\footnote{To be precise, the spanners have a low arboricity.
Recall that the arboricity of a graph is the minimum number of forests it can be partitioned into.}.
In the following, we present the two steps in more detail.

\paragraph*{Step 1: Create a Sparse Spanner $S(G)$.}
    In the first phase, we will construct a spanner $S(G) := (V,S(E))$ of $G$ to reduce the number of edges to $O(n\log n)$ and its arboricity to $O(\log n)$. 
    In particular, we note that $S(G)$ is subgraph of $G$, so every edge in $S(G)$ is a local edge in our hybrid model.
    We will adapt the spanner construction algorithm of Elkin and Neiman \cite{EN18} which in turn is \emph{heavily} influenced by the work of Miller et al. \cite{MPV+15}.   
    The algorithm works as follows: 
    \begin{enumerate}
        \item Each node $v$ independently draws a random value $r_v$ from the exponential distribution with parameter $\beta = \nicefrac{1}{2}$. 
        Values larger than $2\log m$ are discarded.
        \item Each node that did \emph{not} discard its value $r_v$ broadcasts it to all nodes within distance $2\log m+1$. 
        For a simpler presentation of the algorithm, we assume this to be possible for the moment. 
        \item Each node $v \in V$ that received any $r_u$, stores $m_u(v) = r_u - d_G(v,u)$ and the neighbor $p(u)$ from which it first received $u$ (i.e, $p(u)$ is predecessor of $v$ on some path from $u$).
        \item For $v \in V$ let $m(v) := \max \{m_u(v) \mid u \in V \}$.
        Then the spanner's edges are defined as $S(E) := \left\{(v,p_w(v)) \mid m(v) \leq m_w(v) + 1\right\}$.
        \item Last, every node whose degree in $G$ is smaller than $c\log n$ adds \emph{all} of its incident edges to $S(E)$. 
        The constant $c$ is to be determined in the analysis.  
    \end{enumerate}
    Note that the key differences between our algorithm and the counterpart of Elkin and Neiman are that we broadcast the values for only $O(\log m)$ and not $O(\log n)$ rounds, and ignore $r_v$'s that are larger than $O(\log m)$.
    In particular, the algorithm of Elkin and Neiman simply is conditioned on the event that all $r_i$ are small enough.
    However, since our graph's components only have size $m$, and we wish to obtain a runtime proportional to $m$ instead of $n$, we need to pursue a different approach.
    Therefore, our algorithm potentially creates slightly different spanners.
    To illustrate this, consider the case that the node that maximizes $m(v)$ is within distance of $\omega(\log m)$ to $v$.
    Then, our adaptation does not consider this node because we terminate the broadcast before it can reach $v$ (whereas it may have reached $v$ in the original algorithm).
    We compensate this by letting nodes of low degree add all of their edges.
    But since we are only interested in the nodes' outdegrees and not the other properties of the spanner, i.e., the number of edges, this is fine for our case.
     
    In the following, we will show that the resulting graph $S(G)$ is directed and has an outdegree of $O(\log n)$ in expectation and with high probability. 
    We begin with the definition of an \emph{active} node.
    Intuitively, an active node $v \in V$ is a node that is reached by a sufficiently large $r_u$ within $2\log m$ rounds. 
    
\begin{definition}[Active Node]
    Given the non-discarded random values $r_1, \dots, r_{m'}$, we call a node $v \in V$ active, if it holds $m(v) \geq 0$.
    All others are inactive.
\end{definition}

Since for any node $v \in V$ that did not discard its value we have $r_v \geq 0$, and $m_v(v) = r_v - d(v,v) \geq 0$ as $d(v,v)=0$, such a node must be active.
By the same argument, all nodes $w \in V$ that discarded their value only become active, if they receive a value $r_u$ such that $r_u-d(u,w) \geq 0$.
As the following lemma implies, this holds for all nodes whose degree is large enough.
\begin{lemma}
\label{lemma:inactive}
    Let $v \in V$ be an inactive node, then $deg(v) \leq c\log n$, w.h.p., where $deg(v)$ is the degree of $v$ in $G$.
\end{lemma}
\begin{proof}
    Suppose for contradiction that there is any $v \in V$ with $deg(v) > c\log n$ that is inactive.
    
    Let $r_1,\dots,r_{deg(v)}$ be independent random variables sampled from the exponential distribution with parameter ${\frac{1}{2}}$ that are drawn by $v$'s neighbors.
    To show the lemma, we will show that there is at least one $r_i$ with $r_i > 1$.
    We will call such an $r_i$ \emph{good}.
    As $d(v,i)=1$ this implies that $r_i-d(i,v)$ is non-negative, which --- by definition --- is sufficient to show that $v$ is active. 
    The probability for the event that a random value $r_i$ is larger than $1$ is
    \begin{equation}
        \Pr[r_i > 1] = e^{-\frac{1}{2}} \geq \frac{1}{e}.
    \end{equation}
    Thus, the expected number of good values is larger than $\frac{c}{e}\log n$.
    Since the values are drawn independently, a simple application of the Chernoff bound yields that there are at least $\frac{c}{2e}\log n$ good values with high probability.
    
    Note that not all good values are sent to $v$ as they may be discarded.
    Thus, we need to rule out the discarded values.
    We call any $r_i \geq 2\log m$ a \emph{big} value.
    Then, w.h.p., there are at most $\frac{c}{8e}\log n$ big values.
    The proof is straightforward as the probability to draw a big value is
    \begin{equation}
        \Pr[r_i \geq 2\log m] = e^{-\frac{2\log m}{2}} = \frac{1}{m}.
    \end{equation}
    Since there are only $m$ nodes in each component, we have that $deg(v) \leq m$, which implies that the expected number of \emph{big} values (in any neighborhood) is $O(1)$.
    Since the values are drawn independently at random, a simple application of the Chernoff bound yields that there are no more than $\frac{c}{8e}\log n$ big values with probability $1-o(n^{c'})$. Here, $c'$ is a constant that depends on the size $c$.
    
    Combining these two statement yields that there are more good than big values,
    w.h.p., for a large enough $c>16e$.
    Thus, $\Omega(1)$ good values are not discarded.
    Finally, by a union bound, any node with $deg(v)>c\log n$ receives a positive value and is therefore active.
\end{proof}

One immediate implication for an active node is given in the following statement.

\begin{lemma}
For any active node let $u\in V$ be the node maximizing $m_u(v) := r_u -d(u,v)$.
We have that
\[
    d(u,v) < 2 \log m+1.
\]
\end{lemma}
\begin{proof}
    Since any $r_u$ is strictly smaller than $2\log(m)+1$, it holds that $r_u - 2\log m+1 < 0$.
    This rules out any $u$ in distance larger than $2\log m$.
\end{proof}

We now turn to the analysis of our spanner construction.
First, we observe that $S(G)$ is indeed connected, then we show that each node only adds a small number of edges.
For the connectivity, we need the following auxiliary lemma, which corresponds to \cite[Claim 5]{EN18}.

\begin{lemma}\label{lemma:path}
Let $v \in V$ be an active node.
For any $u\in V$, if $v$ adds an edge to $p_u(v)$, then there is a path $P$ between $u$ and $w$ that is fully contained in the spanner $S(G)$.
\end{lemma}
\begin{proof}
We prove the claim by induction on $d(u,v)$ as in \cite{EN18}.

\paragraph*{IB:}First, consider the case that $d(v,u)=1$. 
Since $p_u(x)=u$, the edge $(v,u)$ is added to $S(G)$ and obviously is a path to $u$. 

\paragraph*{IS:} For the step, assume that every active node $w \in V$ with $d_G(u,w)=t-1$ that added an edge to $p_u(w)$ has a path to $u$ in $S(G)$.
Now consider a node $v$ that has $d(u,v)=t$. 
We know that $v$ added an edge to $w=p_u(v)$.
Obliviously, $w$ lies on path to $u$, because $v$ received $u$ via $w$.
Thus, $w$ satisfies $d(u,w)=t-1$. 
It remains to show that this $w$ is active and added an edge to $p_u(w)$. 
We prove these two facts separately:
\begin{description}
    \item[Fact 1: $w$ is active.] 
    Since $v$ is active and added the edge $p_u(v)$, it must hold:
    \begin{equation}
        m_u(v) \geq m(v)-1 \geq -1
    \end{equation}
    This follows because $m(v)$ is non-negative for all active nodes.
    Since $w$ lies on the path from $u$ to $v$, it holds:
     \begin{equation}
        m_u(w) \geq m_u(v)+1 \geq 0
    \end{equation}
    Since per definition $m(w) \geq m_u(w) \geq 0$ the maximum $m(w)$ must be non-negative.
    Thus, $w$ is active.


    
    \item[Fact 2: $w$ added $(w,p_u(w))$.] 
    First we claim that
    \begin{equation}\label{eq:1}
    m(w)\le m(v)+1~.
    \end{equation}
    Seeking contradiction, assume that \eqref{eq:1} does not hold, and let $z \in V$ be the vertex maximizing $m_z(w)$. 
    
    Since $w$ is active, we have $d(z,w)<2\log m+1$, and thus $d(z,v) \leq 2\log m+1$. Hence $v$ will hear the message of $z$. 
    This means that $m_z(v)\geq m_z(w)-1=m(w)-1>m(v)$, which is a contradiction to \eqref{eq:1}.
    
    Recall that $v$ added an edge to $w=p_u(v)$, so by construction
    \begin{equation}\label{eq:2}    
    m_u(v)\geq m(v)-1~.
    \end{equation}
    We conclude that
    \[
    m_u(w)= m_u(v)+1\stackrel{\eqref{eq:2}}{\ge} m(v)-1+1\stackrel{\eqref{eq:1}}{\ge} m(w)-1~,
    \]
    
\end{description}
Thus, $w$ is active and indeed adds an edge to $p_u(w)$, and by the induction hypothesis we are done.
\end{proof}
Then, again very similar to \cite{EN18}, we can show that the resulting spanner is indeed always connected:
\begin{lemma}
    The (undirected version of) graph $S(G)$ is always connected.
\end{lemma}

\begin{proof}
    Consider any edge $(v,w) \in E$. For $S(G)$ to be connected, this edge must either be contained in $S(G)$ or there must be a path connecting $v$ and $w$. 
    We make a case distinction based on whether $v$ or $w$ are active.
    \begin{description}
        \item[Case 1: Either $v$, $w$, or both are inactive.]
        By Lemma \ref{lemma:inactive} the degree of either $v$ or $w$ must be smaller than $c\log n$.
        Then the active node(s) adds all its incident edges to the spanner and --- in particular --- also the edge $(v,w)$ or $(w,v)$ respectively.
        \item[Case 2: Both $v$ and $w$ are inactive.] Let $u$ be the vertex maximizing $m(v)=m_u(v)$, and w.l.o.g assume $m(v)\ge m(w)$. Since $v$ is active, we have that $d_G(u,v) \leq 2\log m$, so $d_G(u,w) \leq 2\log m+1$. Thus, $w$ heard the message of $u$ (which was sent to distance $2\log m+1$). 
        
        This implies that $m_u(w)\ge m_u(v)-1=m(v)-1\ge m(w)-1$, so $w$ adds the edge $(w,p_u(w))$ to $S(G)$. 
        By applying Lemma \ref{lemma:path} on $v$ and $w$, we see that both have shortest paths to $u$ that are fully contained in $S(G)$. 
    \end{description}
\end{proof}
Now we observe the outdegrees of all nodes.
Recall that there are two types of outging edges.
First, every node creates an outgoing edge for all of its predecessors $p_u(v)$ to nodes with $m_u(v) \geq m(v)-1$.
Second, every node of degree lower than $c\log n$, add \emph{all} its edges.
Thus, since the outgoing edges of the second type are naturally bounded by $O(\log n)$ and we only need to consider the first type.

Therefore, we observe for any node $v \in V$ how many values $m_u := r_u - d(u,v)$ are within distance $1$ to $m(v)$, w.h.p.
This directly follows from \cite{EN18}:
\begin{lemma}[Lemma $1$ in \cite{EN18}]\label{lem:MPVX}
Let $d_1\leq\ldots\leq d_{m'}$ be arbitrary values and let $r_1,\dots,r_{m'}$ be independent random variables sampled from the exponential distribution with parameter $\beta$. Define the random variables $M=\max_i\{\delta_i-d_i\}$ and $I=\{i~:~\delta_i-d_i\ge M-1\}$. Then for any $1\le t\le n$,
\[
\Pr[|I|\ge t]= (1-e^{-\beta})^{t-1}~.
\]
\end{lemma}
Now we are able to show the following lemma.
\begin{lemma}
\label{lemma:degree}
    Every node in $S(G)$ has an outdegree of at most $O(\log n)$, w.h.p.
 \end{lemma}
 \begin{proof}
    For all nodes with degree smaller than $c\log n$ the lemma follows immediately.
    Therefore, we only consider nodes of higher degree. 
    Such a node $v \in V$ adds one edge for each $u$ with $m_u(v) \geq m(v)-1$.
    The number $X_v$ of these nodes can be bounded with Lemma \ref{lem:MPVX} to 
    \[
        \Pr[X_v \ge t]= (1-e^{-\beta})^{t-1}~.
    \]
    Thus, by choosing $\beta=\nicefrac{1}{2}$ and $t = \Omega(\log n)$, we see that, w.h.p, no node has more than $t$ values within distance $1$ to its minimum $m(v)$ and only adds $O(\log n)$ edges w.h.p. 
    Thus, by a union bound, every node adds at most $O(\log n)$ edges.   
    This proves the claim.
\end{proof}

It remains to address the algorithm's simplification that the nodes can perform broadcasts in parallel, which could require nodes to forward more than one message over the same edge.
However, Elkin and Neiman observed that it suffices for a node to send the message $(r_u, d_G(u,v))$ for the vertex $u$ that currently maximizes $m_u(v)$ to all of its neighbors. 
Elkin and Neiman further argue that omitting all the other messages will not affect the construction, since if one such message would cause some neighbor of $v$ to add an edge to $v$, then the message about $u$ will suffice, as the latter has the largest $m_u(v)$ value.
For a more detailed account of the implementation, we refer to \cite{EN18}.
    
    \paragraph*{Step 2: Transform $S(G)$ into a bounded degree graph $H$.}
    Now we will construct a bounded degree graph $H$ from $S(G)$.
    Note that $H$ --- in contrast to $S(G)$ --- is \emph{not} a subgraph of $G$ and contains additional edges.
    Although $S(G)$ has few edges in total, there can still be nodes of high degree because there may be nodes with high \emph{indegree}.
    Our goal is that nodes of high indegree redirect their incoming edges to other nodes in order to balance the degrees.
    This technique is conceptually similar to the construction of a child-sibling tree as in \cite{AW07} and \cite{GHSS17}.
    \begin{enumerate}
        \item  
        In the first step, all nodes learn all of their incoming connections in $S(G)$.
        For this, every node $v \in V$ with an edge $e = (v,w)$ in $S(G)$ sends a message containing its identifier to $w$.
        Since $e$ must also have existed in $E$ and each identifier is of size $O(\log n)$,
        this step can be executed in exactly one round in the \CONGEST model.
        \item  Next, we delegate all incoming edges away and create a list of all incoming nodes.
        For the construction, consider a node $v \in V$ and let $N(v) := w_1, \dots, w_k$ be all nodes with $(w_i,v) \in S(G)$, i.e, the incoming edges of $v$.
        W.l.o.g., assume that $w_1, \dots, w_k$ are ordered by increasing identifier.
        Then, for each $i>1$, $v$ sends the identifier of $w_{i}$ to $w_{i-1}$ and vice versa.
        This results in the following set of edges:
        \begin{equation}
        E_H := \bigcup_{v \in V}\bigcup_{i \in [|N(v)|]}  \begin{cases}{}
        \{(v,w_{i}),(w_i,v)\} & i=1\\
        \{(w_{i},w_{i-1}),(w_{i-1},w_{i})\} & i>1\\
        \end{cases}
        \end{equation}{}
    \end{enumerate}{}
    One can easily verify that each node has at most one incoming edge left (i.e., the edge from $w_1$ to $v$) and received at most two edges for \emph{each} outgoing edge (i.e., the edges to $w_{i-1}$ and $w_{i-1}$).
    Thus, the resulting graph $H = (V,E_H)$ has a degree of $O(\log n)$ since each node's outdegree in $S(G)$ is within $O(\log n)$ w.h.p. by Lemma \ref{lemma:degree}.
    
    

   \medskip
Note that both these steps take $O(\log m)$ communication rounds.
The runtime of the first phase only depends on the broadcast of the $r_v$'s and thus takes $O(\log m)$ steps.
In the second step, all nodes only exchange two messages with their neighbors in $S(G)$, so its runtime is $O(1)$.
Since all nodes know the same estimate of $O(\log m)$, the phases can be synchronized via round counters.

\subsection{Spanning Trees} \label{sec:spanning}

We will now show how the algorithm of Theorem~\ref{thm:connected} can be used to construct a spanning tree of the (undirected version of the) initial graph $G$.
For simplicity, we assume that this graph is connected; our algorithm can easily be extended to also compute spanning forests of unconnected graphs by running it in each connected component.
We show the following theorem:

\spanning*

Note that Theorem~\ref{thm:connected} constructs a graph $G_{L'}$ that results from $L' = O(\log n/ \log \log n)$ evolutions of the graph $G_0$ of Lemma~\ref{lem:connected_preprocessing} and that has diameter $O(\log n)$, and degree $O(\log^2 n)$, w.h.p.
First, we construct a spanning tree $S_{L'}$ of $G_{L'}$ by performing a BFS from the node with highest identifier.
Our idea is to iteratively replace all the edges of $S_{L'}$ by edges of $G_{{L'}-1}$, replace these edges by edges of $G_{{L'}-2}$, and so on, until we reach a graph that contains only edges of $G_0$.
We then first break all cycles of this graph using pointer jumping, and finally infer a spanning tree of $G$ by reverting the delegation of edges in Phase 2 of the algorithm of Section~\ref{sec:components}.

More precisely, our algorithm works as follows.
First, the nodes perform a depth-first traversal of $S_{L'}$ using the Euler tour technique.
Specifically, we execute the algorithm of \cite[Lemma 4]{FHS20}.\footnote{Note that the algorithms of \cite{FHS20} can be executed \emph{directly} in our hybrid model.}
As a by-product, the nodes learn the path $P_{L'}$ that corresponds to a depth-first traversal of $S_{L'}$.
This path covers all nodes and, since $S_{L'}$ is a well-formed tree, contains each node at most $O(\log^2 n)$ times.
Next, we want to replace all edges of $P_{L'}$ by edges of $G_0$ in an iterative fashion.

To be able to do that, the two endpoints of every random edge $e$ created throughout the execution of our main algorithm need to know the edges the corresponding token traversed (the edges that \emph{make up} $e$).
Note that the token traverses $\ell = O(\log^2 n)$ nodes.
To annotate each token with the edges it traverses, we need to increase the global capacity of the algorithm of~\ref{thm:connected} to $O(\log^5 n)$, since by Lemma~\ref{lemma:rapid_sampling} each node needs to send and receive $O(\log^3)$ messages, each of which consisting of $O(\log^2 n)$ "submessages".
Therefore, the endpoints of each edge $e$ of $P_{L'}$ can inform the endpoints of all edges of $G_{{L'}-1}$ that make up $e$, which creates a path $P_{{L'}-1}$ that only contains edges of $G_{{L'}-1}$.
In turn, these endpoints of all edges of $P_{{L'}-1}$ can inform all nodes that make up the edge to obtain $P_{{L'}-2}$, and so on, until we obtain a path $P_0$.

\begin{lemma} \label{lem:spanningtree_path}
    $P_0$ contains all nodes of $V$ and can be computed in time $O(\log n)$.
    Furthermore, each node is contained at most $O(\log^4 n)$ times, w.h.p.
\end{lemma}

\begin{proof}
    Since $P_{L'}$ contains all nodes of $V$, and we only repeatedly replace edges by paths, $P_0$ also contains all nodes.
    Furthermore, each node is contained in $P_{L'}$ at most $O(\log n)$ times, w.h.p, since $S_L$ is a well-formed tree.
    Note that for $\Delta = O(\log n)$ Lemma~\ref{lemma:expected_congestion} implies that in each round of each evolution, each node is only traversed by $O(\log n)$ tokens, w.h.p.
    Further, since $\ell = O(\log^2 n)$, during each evolution of the algorithm every node is traversed by a total of $O(\log^3 n)$ tokens.
    Therefore, when replacing an edge of $P_i$ by a path in $G_{i-1}$, each node is only added $O(\log^3 n)$ times, w.h.p.
    Since we have $L' = O(\log n)$, each node is contained at most $O(\log^4 n)$ times in $P_0$, w.h.p.
\end{proof}

To transform $P_0 = (v_1, \ldots, v_k)$ into a spanning tree $S_0$ of $G_0$, each node $v \in V$ selects the edge $e_v$ over which it is reached first in $P$, i.e., $e_v = \{v_{i-1},v_i\}$ such that $v = v_i$ and $i = \text{argmin}_{j \in \{1,\ldots,k\}} v = v_i$.  
These edges can easily be found using \cite[Theorem 1]{FHS20}, which performs pointer jumping and uses the \emph{prefix sum} technique on $P$.
Note that since each node is contained in $P$ at most $O(\log^4 n)$ times by Lemma~\ref{lem:spanningtree_path}, the algorithm can be performed with global capacity $O(\log^4 n)$ in time $O(\log n)$.
The selected edges form a so-called \emph{loop erased path} of $G_{0}$ that covers all nodes, therefore the set $\{e_v \mid v \in V\}$ is a spanning tree $S_0$ of $G_{0}$.

However, $S_0$ may not be a spanning tree of $G$.
Recall that an edge $\{u,w\}$ in $S_0$ may not exist in $G$ (i.e., if it resulted from a redirection of an edge $\{u,v\}$ in $G_0$ in Phase II of Section~\ref{sec:components}, where $u$ and $w$ were incoming nodes of $v$).
However, after computing the edges over which each node is reached first in $P$, we can simply replace each edge $\{u,w\}$ that does not exist in $G$ by the two edges $\{u,v\}$ and $\{w,v\}$ that exist in $G$ using global communication.
Thereby, a node may learn that is is actually reached earlier in $P$, and we "repair" the loop erased path to obtain a spanning tree $S$ of $G$.
We conclude Theorem~\ref{thm:spanning}.

\subsection{Biconnected Components} \label{sec:biconnectivity}

In this section, we present an adaptation of Tarjan and Vishkin's biconnectivity algorithm \cite{bcmain} to compute the biconnected components of $G$ in time $O(\log n)$, proving the following theorem.

\biconnectivity*

The algorithm constructs a \emph{helper graph} $G'=(E,E')$ with the edges of $G$ as nodes and with an edge set $E'$ chosen such that any two edges of $G$ are connected in $G'$ if and only if they lie on a cycle in $G$.
Therefore, the nodes of each connected component of $G'$ are edges of the same biconnected component in $G$.
If there is only one component in $G'$, then $G$ is biconnected.

On a high level, the algorithm can be divided into five steps.
In Step 1, we construct a rooted spanning tree $T$ of $G$ and enumerate the nodes from $1$ to $n$, assigning each node $v$ a  label $l(v)$, according to the order in which they are visited in a depth-first traversal of $T$.
Let $D(v)$ be the set of descendants of $v$ in $T$ (including $v$).
The goal of Step 2 is to compute $nd(v) := |D(v)|$ as well as $high(v) := \max\{l(u) \mid u \in D^+(v)\}$ and $low(v) := \min\{l(u) \mid u \in D^+(v)\}$, where $D^+(v) := D(v) \cup \{u \in V \mid \{u,w\} \in E\setminus T, w \in D(v)\}$ is the union of $v$'s descendants and its descendants neighbors in the undirected version of $G$.
Using these values, in Step 3 the nodes construct the subgraph $G''$ of $G'$ that only contains the nodes that correspond to edges of $T$ (i.e., it does not include nodes for the \emph{non-tree edges} of $G-T$).
The nodes simulate $G''$ in a way that allows them to perform Theorem~\ref{thm:connected} without any overhead to establish a well-formed tree on each connected component of $G''$ in Step 4.
Finally, in Step 5 the components of $G''$ are extended by nodes corresponding to non-tree edges to obtain the full biconnected components of $G$.

In the remainder of this section, we describe how the five steps can be implemented in the hybrid model in time $O(\log n)$ using Theorem~\ref{thm:main} together with the results of \cite{AGG+19} and \cite{FHS20}.
The correctness of Theorem~\ref{thm:biconnectivity} then follows directly from~\cite[Theorem 1]{bcmain}.

\paragraph*{Step 1: Construct $T$.}
$T$ is computed using Theorem~\ref{thm:spanning} in time $O(\log n)$, w.h.p.
The tree can be rooted using the algorithm of~\cite[Lemma 4]{FHS20}, which arranges the nodes of $T$ as an overlay ring that corresponds to a depth-first traversal of $T$ and performs pointer jumping on that ring.
As a by-product, we can easily enumerate the nodes in the order in which they are visited in the depth-first traversal, whereby each node obtains its label.

\paragraph*{Step 2: Compute Subtree Aggregates.}
To retrieve the value $nd(v)$ for each node $v \in V$, the nodes perform the algorithm of \cite[Lemma 6]{FHS20} on $T$:
If each node $u$ stores a value $p_u$, then the algorithm computes the sum of all values that lie in each of $v$'s adjacent subtrees (i.e., the components into which $G$ decomposes if $v$ gets removed) deterministically in time $O(\log n)$; we obtain $nd(v)$ by setting $p_u = 1$ for each $u \in V$.
However, to compute $high(v)$ and $low(v)$, for each node $v \in V$, the nodes need to compute maxima and minima.
Therefore, we need the following lemma, which is a generalization of \cite[Lemma 6]{FHS20}.\footnote{Note that a naive PRAM simulation in a butterfly introduces an additional factor of (at least) $\Theta(\log n)$ to the runtime, which we cannot afford.
Furthermore, this result may be of independent interest for hybrid networks.}

\begin{lemma}\label{lem:bicon:subtreeagg}
    Let $T=(V,E)$ be a tree and assume that each node $v\in V$ stores some value $p_v$.
    Let $f$ be a distributive aggregate function.
    The goal of each node $v$ is to compute the value $f(\{p_w \mid w \in C_u\})$ for each of its neighbors $u$ in $H$, where $C_u$ is the connected component $C$ of the subtree $T'$ of $T$ induced by $V\setminus\{v\}$ that contains $u$. The problem can be solved in time $O(\log n)$, w.h.p.
\end{lemma}

\begin{proof}
    As described before, we enumerate the nodes of $T$ from $1$ to $n$ by assigning them a label $l(v)$ according to the order in which they are visited in a depth-first traversal of $T$ (starting at the node $s$ with smallest identifier).
    Furthermore, we construct a list $L$ as an overlay in ascending order of their label, and root $T$ towards $s$.
    This can be done in time $O(\log n)$ using techniques of \cite{FHS20}.
    Afterwards, the nodes perform pointer jumping on $L$ to create \emph{shortcut edges} $E_S$ for $O(\log n)$ rounds, which decreases the diameter of $L$ to $O(\log n)$.
    Additionally, the endpoints $i$, $j$ of a shortcut edge $\{i,j\} \in E_S$ learn the weight $w(\{i,j\}) := f(\{p_k \mid k \in V, l(i) \le l(k) \le l(j)\})$.
    Now consider some node $v \in V$.
    First, we show how $v$ can compute $f(\{p_u \mid u \in D(v)\})$, i.e., the aggregate of all values in $v$'s subtree.
    Note that this value is exactly $f(\{p_k \mid k \in V, l(v) \le l(k) \le l(w)\})$, where $w$ is the node for which $l(w) = l(v) + |D(v)| - 1$ (i.e., the node in $v$'s subtree with largest label).
    Note that this value is the aggregate of all values on the segment between $v$ and $w$ on $L$.
    To obtain this value, $v$ only needs to learn the weights of at most $O(\log n)$ shortcut edges on that segment.
    More formally, there is a path $P = (v = v_1, v_2, \ldots, v_t = w)$ on $L$ such that $l(v_{k+1}) = l(v_k) + 2^{\lfloor \log(l(w) - l(v_k)) \rfloor}$ for all $k < t$.
    Obviously, $t = O(\log n)$, and there is a shortcut edge between any two consecutive nodes on that path.
    To learn the weights of all these shortcut edges, $v$ needs to contact all $v_k$.
    
    However, since many nodes may want to contact the same node, we cannot send request messages directly, even if each node knew all node identifiers. 
    Instead, we make use of techniques of \cite{AGG+19} to construct \emph{multicast trees} towards each node\footnote{Note that \cite{AGG+19} assumes that the nodes know all node identifiers; however, the nodes on $L$ can easily simulate a butterfly network, which suffices for the algorithms of \cite{AGG+19}.}.
    Since each node needs to contact $O(\log n)$ nodes, it participates in the construction of $O(\log n)$ multicast trees.
    Further, each node $u\in V$ is the root of at most $O(\log n)$ multicast trees (one for each of its adjacent shortcut edges).
    When $u$ multicasts the weight of a shortcut edge in the respective multicast tree, all nodes that participated in the construction of that tree will be informed.
    Plugging the parameters $L = O(n \log n)$ (which is the total number of requests) and $l = \hat{l} =  O(\log n)$ (which is the number of weights each node wants to learn) into \cite[Theorem 2.3]{AGG+19} and \cite[Theorem 2.4]{AGG+19}, we get that each node learns all weights in time $O(\log n)$, w.h.p.\footnote{\cite[Theorem 2.4]{AGG+19} actually restricts each node to act to multicast at most \emph{one} multicast message.
    However, the theorem can easily be extended to allow multiple messages without increasing the runtime in our case.}
    
    After having learned the weights of all edges on $P$, $v$ can easily compute $f(\{p_u \mid u \in D(v)\})$.
    By sending this value to its parent in $T$ (over a local edge), each node learns the aggregate of the subtree of each of its children.
    It remains to compute $f(\{p_u \mid V \setminus D(v)\})$, i.e. the aggregate of all \emph{non-descendants} of $v$.
    Note that since the descendants of $v$ form a connected segment from $v$ to $w$ in $L$, these non-descendants form exactly two segments on $L$: one from $s$ to $v$ (excluding $v$), and one from $w$ to the last node of $L$ (excluding $w$).
    Using the same strategy as before, $v$ can compute the aggregate of all these values by learning the weight of $O(\log n)$ shortcut edges.
\end{proof}

\paragraph*{Step 3: Construct $G''$.}
Recall that $G''$ is the subgraph of $G'$ induced only by the nodes that correspond to edges of $T$.
In order to simulate $G''$, we let each node $v$ of $G$ act on behalf of the node of $G''$ that corresponds to $v$'s parent edge.
That is, when simulating an algorithm on $G''$, $v$ is responsible for all messages the node corresponding to $v$'s parent edge is supposed to communicate.
We now need to connect all nodes corresponding to edges that are on a common simple cycle in $G$. 
Tarjan and Vishkin showed that it suffices to consider the simple cycles consisting of a nontree edge and the unique shortest path between its adjacent nodes \cite{bcmain}. To do this, they propose the following rules:
\begin{enumerate}
    \item If $(v,u)$ and $(w,x)$ are edges in the rooted tree $T$ (directed from child to parent), and $\{v,w\}$ is an edge in $G-T$ such that $v$ is no descendant of $w$ and $w$ is no descendant of $v$ in $T$ (i.e., $v$ and $w$ lie in different subtrees), add $\{\{u,v\},\{x,w\}\}$ to $G''$.
    \item If $(w,v)$ and $(v,u)$ are edges in $T$ and some edge of $G$ connects a descendant of $w$ with a non-descendant of $v$, add $\{\{u,v\},\{v,w\}\}$ to $G''$. 
\end{enumerate}
Roughly speaking, for each non-tree edge $\{v,w\}$ that connects two different subtrees of $T$, the first rule connects the parent edges of $v$ and $w$, whereas the second rule connects all edges of $T$ that lie on the two paths from $v$ to $w$ to their lowest common ancestor.
An illustration of these rules can be found in the left and center image of Figure~\ref{fig:bcrules}. 

\begin{figure}
    \begin{center}
\tikzset{
    treenode/.style = {circle,draw,fill=lightgray,thick,minimum size=1.8em},
    edgenode/.style = {midway,inner sep=0pt},
    treeedge/.style = {draw,<-,very thick}
}
\begin{tikzpicture}
    \node[treenode] (root) {};
    \node[treenode,below left=of root] (l1) {};
    \node[treenode,below right=of root] (r1) {};
    \node[treenode,below=of l1] (l2) {u};
    \node[treenode,below=of r1] (r2) {x};
    \node[treenode,below=of l2] (v) {v};
    \node[treenode,below=of r2] (w) {w};
    
    \draw[treeedge] (root) -- (l1) node[edgenode] (rl1) {};
    \draw[treeedge] (l1) -- (l2) node[edgenode] (l1l2) {};
    \draw[treeedge] (l2) -- (v) node[edgenode] (l2v) {};
    \draw[treeedge] (root) -- (r1) node[edgenode] (rr1) {};
    \draw[treeedge] (r1) -- (r2) node[edgenode] (r1r2) {};
    \draw[treeedge] (r2) -- (w) node[edgenode] (r2w) {};

    \draw[thick] (v) -- (w) node[edgenode] (vw) {};

    \draw[ultra thick] (l2v) to [bend left] (r2w);

\end{tikzpicture}
\tikzset{
    treenode/.style = {circle,draw,fill=lightgray,thick,minimum size=1.8em},
    edgenode/.style = {midway,inner sep=0pt},
    treeedge/.style = {draw,<-,very thick}
}
\begin{tikzpicture}
    \node[treenode] (root) {u};
    \node[treenode,below left=of root] (l1) {v};
    \node[treenode,below right=of root] (r1) {};
    \node[treenode,below=of l1] (l2) {w};
    \node[treenode,below=of r1] (r2) {};
    \node[treenode,below=of l2] (v) {};
    \node[treenode,below=of r2] (w) {};
    
    \draw[treeedge] (root) -- (l1) node[edgenode] (rl1) {};
    \draw[treeedge] (l1) -- (l2) node[edgenode] (l1l2) {};
    \draw[treeedge] (l2) -- (v) node[edgenode] (l2v) {};
    \draw[treeedge] (root) -- (r1) node[edgenode] (rr1) {};
    \draw[treeedge] (r1) -- (r2) node[edgenode] (r1r2) {};
    \draw[treeedge] (r2) -- (w) node[edgenode] (r2w) {};

    \draw[thick] (v) -- (w) node[edgenode] (vw) {};

    \draw[ultra thick] (l2v) to [bend right=60] (l1l2);
    \draw[ultra thick] (l1l2) to [bend right=60] (rl1);
    \draw[ultra thick] (r2w) to [bend left=60] (r1r2);
    \draw[ultra thick] (r1r2) to [bend left=60] (rr1);

    \draw[ultra thick,dashed] (l2v) to [bend left] (r2w);

\end{tikzpicture}
\tikzset{
    treenode/.style = {circle,draw,fill=lightgray,thick,minimum size=1.8em},
    edgenode/.style = {midway,inner sep=0pt},
    treeedge/.style = {draw,<-,very thick}
}
\begin{tikzpicture}
    \node[treenode] (root) {};
    \node[treenode,below left=of root] (l1) {};
    \node[treenode,below right=of root] (r1) {};
    \node[treenode,below=of l1] (l2) {};
    \node[treenode,below=of r1] (r2) {u};
    \node[treenode,below=of l2] (v) {v};
    \node[treenode,below=of r2] (w) {w};
    
    \draw[treeedge] (root) -- (l1) node[edgenode] (rl1) {};
    \draw[treeedge] (l1) -- (l2) node[edgenode] (l1l2) {};
    \draw[treeedge] (l2) -- (v) node[edgenode] (l2v) {};
    \draw[treeedge] (root) -- (r1) node[edgenode] (rr1) {};
    \draw[treeedge] (r1) -- (r2) node[edgenode] (r1r2) {};
    \draw[treeedge] (r2) -- (w) node[edgenode] (r2w) {};

    \draw[thick] (v) -- (w) node[edgenode] (vw) {};

    \draw[ultra thick,dashed] (l2v) to [bend right=60] (l1l2);
    \draw[ultra thick,dashed] (l1l2) to [bend right=60] (rl1);
    \draw[ultra thick,dashed] (r2w) to [bend left=60] (r1r2);
    \draw[ultra thick,dashed] (r1r2) to [bend left=60] (rr1);

    \draw[ultra thick,dashed] (l2v) to [bend left] (r2w);

    \draw[ultra thick] (vw) to [bend left] (r2w);
\end{tikzpicture}
    \end{center}
    \caption[The three rules for the construction of $G'$.]{The directed edges are tree edges, and the undirected edge is a non-tree edge.
    Left: The first rule adds an edge between the two parent edges of $v$ and $w$.
    Center: The second rule connects all nodes on the two paths from $v$ to $w$ to their lowest common ancestor.
    Right: The edge $\{v,w\}$ is connected to the component using the third rule.}
    \label{fig:bcrules}
\end{figure}
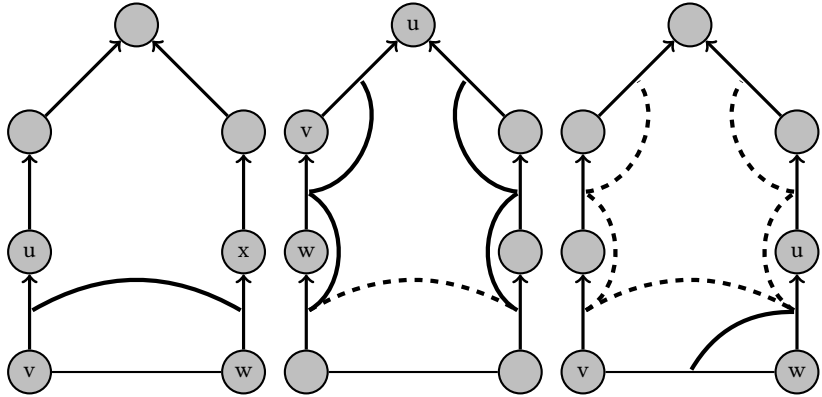

As Tarjan and Vishkin point out, each node $v$ can determine each connection of its parent edge that is formed according to the first rule by comparing $l(v) + nd(v)$ with the label $l(u)$ of each of its neighbors $u$ in $G$; if $l(v) + nd(v) \le l(u)$, then the two parent edges of $v$ are connected in $G''$.
For the second rule, each node $v$, $l(v) \neq 1$ with child $w$ connects its parent edge with the parent edge of $w$ if $low(w) < v$ or $high(w) \ge v + nd(v)$.

\paragraph{Step 4: Compute Connected Components of $G''$.}
To compute the connected components of $G''$, we execute the algorithm of Theorem~\ref{thm:connected} on $G''$.
Note that every two nodes that are connected in $G''$ are simulated by adjacent nodes in $G$; therefore, the local communication in $G''$ can be carried out using the local edges of $G$.
Furthermore, since each node of $G$ simulates at most one node of $G''$, the global communication can also be simulated with the same communication capacity as in Theorem~\ref{thm:connected}.
After $O(\log n)$ rounds, w.h.p., we have established a well-formed tree on each connected component of $G''$.

\paragraph{Step 5: Extend $G''$ to $G'$.}
Finally, we incorporate the non-tree edges into the connected components of $G''$ using the following rule of Tarjan and Vishkin.
\begin{enumerate}
    \setcounter{enumi}{2}
    \item If $(w,u)$ is an edge of $T$ and $\{v,w\}$ is an edge in $G-T$, such that $l(v)<l(w)$, add $\{\{u,w\},\{v,w\}\}$ to $G''$.
\end{enumerate}
An example can be found in the right image of Figure \ref{fig:bcrules}.
Note that this only extends the connected components of $G''$ by single nodes (i.e., it does not merge components of $G''$).
Therefore, afterwards we know the biconnected component of each edge of $G$.
Specifically, if there is only one biconnected component in $G'$ (which can easily be determined by counting the number of nodes that act as the root of a well-formed tree in $G''$) we can determine whether $G$ is biconnected.
Furthermore, we can determine the cut nodes and bridge edges in $G$.
We conclude Theorem~\ref{thm:biconnectivity}.

\subsection{Maximal Independent Set} \label{sec:mis}
Finally, we describe our Maximal Independent Set (MIS) algorithm.
Recall that in the MIS problem, we ask for a set $S \subseteq V$ such that (1) no two nodes in $S$ are adjacent in the initial graph $G$ and (2) every node $v \in V \setminus S$ has a neighbor in $S$.
By a result of Kuhn, Moscibroda and Wattenhofer \cite{KMW04}, there are graphs of degree $d$ in which computing the MIS takes $\Omega({\frac{\log d}{\log\log d}})$ rounds, even in the \LOCAL model.
In models in which the communication graph is much tighter (which roughly corresponds to our notion of \emph{global communication}), the runtime is often \emph{exponentially} better: for example, both in the congested clique and the MPC model \cite{GGJ20,BBD+19,GGKMR18,BFU19} one can achieve a runtime of $O(\log\log n)$.
Many state-of-the-art MIS algorithms employ the so-called shattering technique \cite{BEPSS16,Gha16}, which conceptually works in two stages\footnote{Note that the faster algorithms are more intricate and use more preprocessing stages to reduce degrees, but still rely on this scheme.}:
First, there is the so-called shattering stage, where the problem is solved for the majority of nodes using a local strategy.
As result of this stage, each nodes knows --- with probability $1-o(d)$ --- whether it is in the MIS or has a neighbor in the MIS.
This implies that each undecided node has in expectation less than one undecided neighbor.
Thus, by a Galton-Watson argument, the graph is \emph{shattered} into small isolated subgraphs of undecided nodes.
Then, in the second stage, the MIS is solved on these subgraphs.
In models with massive global communication, all remaining nodes and edges of a component are gathered at single node using the global communication and then solved locally. 
This, of course, requires this node to receive a huge amount of messages in a single round.
Because of this high message load, this approach cannot directly be used in our model.
However, we can do something similar that requires far less messages while still coming close to the $\Omega({\frac{\log d}{\log\log d}})$ bound for \LOCAL.
This emphasizes that even a small amount of non-local communication is as strong as unbounded local communication. 
More precisely, we prove the following theorem.

\mis*

Before we go into the details of our algorithm, we take a short detour to the \CONGEST model.
Here, the MIS problem can be solved in time $O(\log n)$, in expectation and w.h.p., due to a celebrated algorithm by Luby \cite{Luby86} and Alon et al. \cite{ALI86}.
The idea behind the algorithms is quite simple:
Each node picks a random rank in $[0,1]$ which is sent to all neighbors.
Then, all local minima join the MIS and inform their neighbors about it.
All remaining nodes, i.e., nodes that did not join the set and have no neighbor that joined the set, repeat this process until every node has decided.
Later, in \cite{MRNZ09} Métivier et al. provided a simpler analysis, which shows that it is actually sufficient to send a single bit per round and edge.

For our algorithm, we take a closer look at the fact that Métivier's algorithm has an expected runtime of $O(\log n)$.
In particular, it holds that in every round in expectation half of all edges disappear due to nodes deciding (see \cite{MRNZ09} or the appendix of \cite{Gha16} for a comprehensive proof).
Thus, if we execute it on a subgraph with $m^2$ edges, where $m^2 << n^2$, it finishes after $O(\log m)$ rounds in expectation.
That means, by Markov's inequality, with at least constant probability, the algorithm actually only takes $O(\log m)$ rounds.
Therefore, if we execute it $O(\log n)$ times independently in parallel, there must be at least one execution that finishes within $O(\log m)$ rounds, w.h.p.

Now, again, observe the MIS framework using the shattering technique and consider the undecided nodes after the shattering stage.
Instead of reporting all edges to an observer that solves the problem locally for each subgraph of undecided nodes, the nodes can simply report to this observer when their executions finish.
Once there is one execution in which all nodes finished, the observer signals the nodes to stop via broadcast and also tells them which execution finished.
To do so efficiently, we execute the algorithm of Theorem~\ref{thm:connected} on each component of undecided nodes and let the root of each established well-formed tree act as the observer.

More precisely, our algorithm to solve the MIS problem operates in the following three steps of length $O(\log d + \log\log n)$ each. 
To synchronize these steps, we need to assume that, in addition to $\log L$ as an approximation of $\log \log n$, the nodes know an approximation of $O(\log d)$.

\paragraph*{Step 1: Shatter the Graph into Small Components.}
First, we run Ghaffari's (Weak-)MIS algorithm from \cite{Gha16} for $O(\log d)$ rounds.
Let $G_1, \dots, G_k$ be the connected components of $G$ that only consist of undecided nodes (obviously, the nodes can use the local edges to determine which of its neighbors are in the same component).
The remainder of our algorithm will run on each of these $G_i$'s in parallel.

\paragraph*{Step 2: Construct an Overlay for each Component.}
Next, we establish a well-formed  tree $S_i$ on each $G_i$ using the algorithm of Theorem~\ref{thm:connected}.

\paragraph*{Step 3: Execute Métivier's Algorithm in Parallel.}
Finally, we construct an MIS for each $G_i$ as follows:
\begin{enumerate}
    \item On each $G_i$, we run the MIS algorithm of Métivier et al. independently $\Theta(\log n)$ times in parallel.
    Since each execution only needs messages of size $1$, this can be done in the \CONGEST model.
    More precisely, the nodes simply send random bit strings of length $O(\log n)$, where the $i^{th}$ bit belongs to execution $i$.
    \item  Whenever an execution $i$ finishes on a node $v \in V_i$, i.e., a node or one of its neighbors joins the MIS, it uses $S_i$ to send a message to the root that contains the execution and the current round.
    Since there are at most $O(\log n)$ executions finishing in each given round, the information on which executions have finished can be fitted into $O(\log n)$ bits.

    \item The root broadcasts all finished executions to the nodes using $S_i$.
    \item The nodes adopt the result of the first execution that finishes.
    If several executions finish simultaneously, the lexicographically smallest one is chosen.
\end{enumerate}

We are now able to prove Theorem \ref{thm:mis}.

\begin{proof}[Proof of \autoref{thm:mis}]
    First, note that Ghaffari's algorithm can seamlessly be implemented in the \CONGEST model as it only sends $O(\log n)$ sized messages.
    After executing it, each knows with probability $1-o(\Delta)$ whether it is in the MIS.
    Furthermore, the random decision only depends on a node's $2$-neighborhood.
    Thus, w.h.p, the graph is \emph{shattered} into isolated, undecided components $G_1, \dots, G_k$ of size at most $O(d^4\log_{d} n)$ (see, e.g., \cite[Lemma 4.2, (P2)]{Gha16}).
    
    Now consider the construction of well-formed trees $S_1, \dots, S_k$ for these components.
    Since each component has size $O(d^4\log_{d} n)$, the construction takes time $O(\log(d^4\log_{d} n)) = O(\log d + \log \log n)$, w.h.p., by Theorem~\ref{thm:connected}.
    Further, the resulting trees $S_1, \dots, S_k$ have a height of $O(\log d + \log \log n)$.
    This allows us to (deterministically) compute aggregate functions on each $S_i$ in time $O(\log d + \log \log n)$.
    
    Now consider the last step and fix a component $G_i$ with its corresponding tree $S_i$.
    Let $j$ be the be the index of the first successful execution of Métivier's algorithm, i.e., the first execution where all nodes have either joined the MIS or have a neighbor that joined.
    Then, after $O(\log d + \log \log n)$ rounds, the root of $S_i$ is aware of index $j$ through a simple aggregation.
    The root then broadcasts $j$ to all nodes in $S_i$.
    Thus, after another $O(\log d + \log \log n)$ rounds, all nodes are aware of $j$ and stop. 
    
    Finally, we observe that the algorithm of Theorem~\ref{thm:connected} requires a global capacity of $O(\log^3 n)$, which dominates the required global capacity of all the other algorithms.
    The theorem follows.
\end{proof}

\section{Concluding Remarks and Future Work}
\label{sec:future}

In this paper, we answered the following longstanding open question: \emph{Can an overlay network of polylogarithmic degree be transformed into a graph of diameter $O(\log n)$ in time $O(\log n)$ with polylogarithmic communication?}
Whereas our solution is asymptotically time-optimal, our communication bounds may likely be improved.
As pointed out in Section~\ref{sec:problem}, if the initial degree is $d$, then our nodes need to be able to communicate $\Theta(d \log n)$ many messages.
For constant degree, $O(\log n)$ messages suffice, i.e., the algorithm works in the NCC$_0$.
However, as is implicitly proposed in~\cite{AAC+05}, there might be an algorithm that only requires a communication capacity of $\Theta(d)$.
Eradicating the additional $\log n$ factor from our algorithm seems to be non-trivial and poses an interesting goal.


\subsection{Churn-resistent Overlay Construction}

As mentioned earlier, another possibly interesting application of our algorithm is the construction of robust overlay networks under churn, i.e., nodes joining a leaving the network during the construction.
Here, one promising approach is to ensure that our algorithm maintains a sufficiently high vertex expansion throughout every evolution.
That means, each subset of nodes must not only have many edges that lead out of the subset, but also must be connect to many \emph{different} nodes to handle a big fraction of nodes leaving.
Thus, we must additionally analyze how many random walk tokens emitting from a given subset end at the same node.
This likely can be done by using more advanced spectral and/or combinatorical methods.

To illustrate this claim, assume every $\ell$ rounds a, say, logarithmic fraction of the nodes fail and thereby drop all tokens they received.
This, of course, is a very basic churn model, but it is sufficient to convey our point as other churn models are usually stronger and we would face a similar problem.
If each nodes fails independently with probability $p := \frac{1}{\Delta}$, one can easily verify that the \emph{expected} number of outgoing edges in each phase is only decreased by $O(1)$.
Thus, in expectation, the algorithm continues to increase the graphs conductance by a constant factor each phase.
However, this result does not follow with high probability anymore as the individual walks that create the connections do not fail independent of one another. 
In the extreme case, i.e., if all walks of a single end end at the same failing node, the corresponding node does not create a single edge with probility $p >> \frac{1}{n^c}$.
To mitigate this, we need how many random walks (of a single node) end at distinct nodes.
If a constant fraction of walks ends at different nodes, then a constant number of the tokens is dropped independently.
This is enough for the Chernoff bound to kick in, as we could now show that constant fractions of tokens survives w.h.p.
(that means a constant fraction of the independent tokens).

\subsection{Property Testing}
Further, our algorithm can be used as a basis for conductance testing in hybrid models.
In conductance testing, for given parameters $\Phi$, $c\Phi^{2+o(1)}\log n$ and $\epsilon$, the goal is to accept that have conductance $\Phi$ and reject graphs that are at least $\epsilon$-far from having conductance $c\Phi^{2+o(1)}\log n$. 
Here, the term $\epsilon$-far means that an $\epsilon$-fraction of all edges must be changed to obtain 
a graph of the desired conductance.
In this particular scenario, it suffices to find a subset of size $\Theta(\epsilon n)$ with bad conductance, i.e., $O(c\Phi^{2+o(1)}\log n)$, to reject the graph.
There are two algorithms \cite{fichtenberger2018two} or \cite{LMOS20} that consider this problem in distributed models.
For these algorithms, either the runtime or the global communication are in $\widetilde{\Theta}(\Phi)$.
Our algorithm --- most likely --- can reduce this because we can reduce the conductance in a preprocessing step.
Any graph of conductance $\Theta(\Phi)$ can be transformed into a constant conductance graph in $O(\frac{\log(\Phi)}{\log(\ell)})$ rounds.
On the other hand, if we apply our algorithm on graph with a set with conductance $O(\Phi^{2+o(1)}\log n)$, we observe that in every round the conductance can only increase by a factor of $2\ell$ as at most $2\ell$ random walk enter or leave the subset in expectation.
Thus, after $O(\frac{\log(\Phi)}{\log(\ell)})$ rounds, if we choose $\ell$ careful enough, the set likely still has a somewhat bad conductance.
Then, one can apply the algorithm of \cite{fichtenberger2018two} or \cite{LMOS20} with reduced complexity on the graph created by our algorithm. 
Since the exact bounds for $\Phi$, $\Phi^{2+o(1)}\log n$ and $\epsilon$ rely on a very careful choice of $\ell$, we defer this application to future work.

\subsection{Minimum Spanning Trees}
Whereas our algorithm can be used to quickly compute spanning trees, we do not know whether our techniques can also help in finding minimum spanning trees.
There does not seem to be any reason to believe that computing an MST is inherently harder; however, one might need much more sophisticated techniques.
Whereas the MST algorithms for more powerful models such as the congested clique or the MPC model~\cite{Now19, JN18, GP16} hardly seem applicable, it might be worthwhile to investigate whether PRAM algorithms provide useful techniques for overlay networks~\cite{PR02,CHL01}.
Furthermore, coming up with \emph{deterministic} algorithms for this problem seems to be even harder~\cite{AS87}.

\subsection{Another Possibility for an $O(\log(n))$-time construction} \label{sec:halperin}
We point out, however, that it \emph{may} be possible to adapt the algorithm of Halperin and Zwick \cite{HZ01} to our model, which uses $O(n+m)$ processors to construct a spanning tree in time $O(\log{n})$.
Similar to \cite{AAC+05,AW07,GHSS17,GHS19}, the idea of the algorithm is to repeatedly merge supernodes.
To merge a sufficiently large set of supernodes at once, the authors observe that it suffices to perform $O(\log n)$ random walks of length $\ell$ to discover $\ell^{1/3}$ many supernodes, w.h.p.
As in the PRAM model, it is possible to perform such random walks in time $O(\log \ell)$ in overlay networks \cite{DGS16,AS18} under certain conditions.
If the initial graph is $d$-regular, and we allow a node capacity of $\Theta(d\log n)$, we believe that the algorithm of \cite{HZ01}, together with the algorithm of \cite{AW07}, can be applied to our model to construct a low-diameter overlay in time $O(\log n)$.
However, this adaption is highly non-trivial and the resulting algorithm will be significantly more complex than our solution.
Further, our algorithm has the advantage that it's runtime is closely tied to the graph's conductance, which makes it much faster on graphs that already have a conductance of $o(n^\varepsilon)$.
Also, our algorithm can be used for other problems that depend on the graphs conductance, e.g., property testing.

\clearpage

\bibliographystyle{plain}
\bibliography{main}

\end{document}